%% file: main.tex
\documentclass{ieeecolor}
\usepackage{generic}
\usepackage{cite}
\usepackage{amsmath,amssymb,amsfonts}
\usepackage{algorithmic}
\usepackage{graphicx}
\usepackage{algorithm,algorithmic}
\usepackage{hyperref}
\usepackage{textcomp}

\usepackage[utf8]{inputenc} 
\usepackage[T1]{fontenc}        
\usepackage{url}           
\usepackage{booktabs}            
\usepackage{nicefrac}       
\usepackage{microtype}      
\usepackage[compatibility=false]{caption}
\usepackage{subcaption}
\usepackage{nicematrix}

\usepackage{float}
\usepackage{dsfont}
\usepackage{bm}
\usepackage{mathtools}

\usepackage{xurl}
\usepackage{hyperref}
\usepackage{xspace}
\newtheorem{assumption}{Assumption}
\newtheorem{definition}{Definition}
\newtheorem{lemma}{Lemma}
\newtheorem{theorem}{Theorem}
\newtheorem{remark}{Remark}
\newtheorem{proposition}{Proposition}
\newtheorem{example}{Example}
\newtheorem{corollary}{Corollary}

\usepackage{paralist}

\newcommand{\R}{\mathbb R}
\newcommand{\F}{\mathcal F}
\newcommand{\Z}{\mathcal Z}
\newcommand{\D}{\mathcal D}

\newcommand{\Nvz}{\mathcal{N}_{v,z}}
\newcommand{\Sb}{\mathbb{S}}
\newcommand{\Pb}{\mathbb{P}}
\newcommand{\Eb}{\mathbb{E}}

\makeatletter
\newcommand{\vast}{\bBigg@{4}}
\newcommand{\Vast}{\bBigg@{5}}
\makeatother

\newif\ifshowcomments
\showcommentstrue 

\showcommentsfalse 

\newcommand{\ThetaSet}{\mathcal{U}_{\Theta}^{T}} 

\makeatletter
\def\ps@titlepagestyle{%
  \def\@oddfoot{}\def\@evenfoot{}%
  \def\@oddhead{\hbox{}\scriptsize\leftmark\hfil\thepage}%
  \def\@evenhead{\scriptsize\thepage\hfil\leftmark\hbox{}}%
}
\makeatother

\def\BibTeX{{\rm B\kern-.05em{\sc i\kern-.025em b}\kern-.08em
    T\kern-.1667em\lower.7ex\hbox{E}\kern-.125emX}}

\begin{document}


\title{Finite Sample Identification of    Analytic \\Nonlinear Systems}
\input{authors}

\maketitle

\input{Abstract}
\input{Introduction}

\input{Problem_Formulation}
\input{Main_Results}
\input{proof_main}

\input{SME}

\input{CounterExample}

\input{Numerical_Examples}

\input{Conclusion}

\appendices
\renewcommand{\thesubsection}{\Alph{subsection}}
\input{appendix_1}

\input{appendix_2}

\input{appendix_6}

\section*{References}

\vspace{-20pt}
\bibliographystyle{IEEEtran}
\bibliography{references}

\end{document}

%% file: authors.tex
\author{Negin Musavi, 
Ziyao Guo, 
Geir E. Dullerud, 
and Yingying Li
\thanks{Negin Musavi, Ziyao Guo and Yingying Li are with University of Illinois at Urbana-Champaign.  Geir E. Dullerud is with University of Minnesota, Minneapolis. (nmusavi2@illinois.edu, ziyaog2@illinois.edu, dullerud@umn.edu, yl101@illinois.edu)}}

%% file: Abstract.tex
\begin{abstract}
This paper studies the identification of linearly parameterized nonlinear (LPN) systems. Although LPN systems share the same linear parameterization structure as linear systems, they are  more challenging to identify. In particular, previous work has shown, through a counterexample based on a piecewise-affine system, that non-active exploration is generally insufficient for LPN system identification. In this paper, we consider  LPN systems with real-analytic feature functions. We show that non-active exploration is sufficient for the identification of this class of systems  by establishing  non-asymptotic convergence rates of least-squares estimation and  set-membership estimation.  In addition, we provide counterexamples to show that non-active exploration may not be sufficient for system identification for non-real-analytic systems, even if those systems are infinitely differentiable.  We present numerical experiments to further support and validate our theoretical results.
\end{abstract}

\begin{IEEEkeywords}
System identification. Non-asymptotic analysis.  Least Squares Estimation. Linearly parametrized nonlinear systems

\end{IEEEkeywords}

%% file: Introduction.tex
\section{Introduction}\label{sec:intro}
The past decade has witnessed significant interest in finite sample analysis of system identification via a statistical learning perspective. For example, the sample complexity of linear system identification has been well studied \cite{simchowitz2020naive,simchowitz2018learning,li2023non,oymak2019non,faradonbeh2018finitetac,tsiamis2019finite,ziemann2023tutorial,faradonbeh2018finiteauto}, and there are many efforts trying to generalize these results to nonlinear system identification \cite{mania2022active,foster2020learning,khosravi2023representer,kowshik2021near,lee2024active,ziemann2022single,sattar2022non}. A natural and popular direction is the identification of linearly parameterized nonlinear (LPN) systems:
\begin{equation*}\label{equ: linearly parametrized}
	x_{t+1}=\Theta_* \phi(x_t, u_t)+w_t
\end{equation*}
where $\Theta_*$ is a matrix of unknown parameters and $\phi(x_t, u_t)$ is a known vector of  nonlinear feature functions. LPN systems can be viewed as a direct generalization of linear systems and enjoy wide applications, e.g. robotics \cite{siciliano2010robotics,alaimo2013mathematical}, power systems \cite{simpson2016voltage}, transportation \cite{kong2015kinematic}, etc.

For the finite sample analysis of LPN systems, there are  two main lines of research. The first line of research focuses on bilinear systems \cite{sattar2025finite,chatzikiriakos2026end,sattar2022finite,sattar2025learning}. Interestingly, bilinear systems enjoy similar performance of linear systems, for example, bilinear systems can also be learned by least square estimation (LSE) under i.i.d. random inputs (non-active exploration) and achieve a similar  convergence rate of $O(1/\sqrt T)$ \cite{sattar2022finite,sattar2025learning,sattar2025finite}. The second line of research focuses on identification with active exploration  because it has  been shown that, unlike bilinear systems,  general LPN systems may not be efficiently learned by non-active explorations. In particular, \cite{mania2022active} provides a counter-example of a piecewise-affine system and shows that i.i.d. random inputs are insufficient to identify the system under bounded process noises. Motivated by this, various active exploration methods have been proposed for LPN systems and their finite sample performance have been analyzed \cite{mania2022active,kowshik2021near,khosravi2023representer,chatzikiriakos2026end,lee2024active}. 

However, there is a notable gap between bilinear systems, which are infinitely differentiable, and piecewise-affine systems, which are non-differentiable. This leads to an intriguing question:
\textit{Q: to what extent can non-active exploration still work for LPN systems?}

Beyond the theoretical interest in bridging this gap in the finite-sample analysis of system identification, non-active exploration also offers  practical value for online learning and identification with single-trajectory data,  where the control inputs should simultaneously accomplish  control objectives and facilitate exploration of the unknown system. This tension between control performance and information acquisition is commonly referred to as the exploration–exploitation tradeoff in learning literature \cite{sutton2018reinforcement} and is closely related to the notion of dual control in  control literature \cite{bar1974dual}.
Non-active exploration provides a clean solution to address this challenge: by extending the i.i.d. random inputs to randomly perturbed control policies $u_t=\pi(x_t)+\eta_t$, one can choose the policy $\pi(\cdot)$ based on the control objectives/tasks, such as model predictive control \cite{rawlings2017model} or linear quadratic regulator under model estimation \cite{mania2019certainty}, and utilize the random noise $\eta_t$ to provide sufficient exploration for system identification. In this way, it could balance exploration and exploitation together instead of allocating a long period of time on pure exploration as in active exploration methods.

\vspace{4pt}

\noindent\textbf{Contributions.} Our contributions are four-fold.
\begin{itemize}
    \item Our first contribution is showing that non-active exploration is sufficient for LPN systems under two key conditions: real-analytic feature functions and semi-continuous noises. We establish non-asymptotic convergence rates $O(1/\sqrt T)$ for LSE in this setting, which is the same as the rates of linear and bilinear system identification.  We also provide practical examples for these two conditions.
    \item In addition to LSE, we  consider another popular system identification method in control, set-membership estimation \cite{fogel1982value,lu2023robust,li2024icml,lorenzen2019robust}, and analyze its non-asymptotic convergence rate under the stochastic setting. 
    \item We further provide  counter-examples to show the importance of the real-analytic condition:  for non-real-analytic systems, even if they  are finitely differentiable or infinitely differentiable, non-active exploration may still not be sufficient for system identification.
    \item Finally, we provide numerical results to compare our theoretical convergence rates of LSE and SME with their corresponding empirical convergence rates. We further compare the numerical performance of LSE, SME, and the active exploration in \cite{mania2022active}, and a fast SME algorithm inspired by  \cite{lorenzen2019robust,lu2019robust}. 
\end{itemize}

\vspace{4pt}

\noindent\textbf{Related Work.} 
As a subclass of LPN systems, bilinear system have been extensively studied recently \cite{sattar2025finite,chatzikiriakos2026end,sattar2022finite,sattar2025learning}, which is largely motivated by Koopman transform of control-affine systems \cite{williams2015data}. Both active and non-active exploration have been studied with provable non-asymptotic convergence rates, and there are also extensions to partially observable cases \cite{sattar2025finite} and end-to-end control design \cite{chatzikiriakos2026end}, etc. 

Inspired by neural network parameterization, nonlinear systems of the form $x_{t+1}=\phi(A_* x_t)+w_t$ is also studied in the literature, where $\phi(\cdot)$ is a known nonlinear link function and $A_*$ is unknown. The least square cost is no longer quadratic or even convex in this case and various optimization methods have been proposed to learn this type of systems \cite{kowshik2021near,sattar2022finite,foster2020learning}.

Another related  line of research focuses on nonlinear regression with dependent data \cite{ziemann2022learning,ziemann2023tutorial,ziemann2024sharp},\footnote{$y_t=f_*(x_t)+w_t$ is considered, where $x_t$ and $y_t$ correlate with the historical data.} which can be applied to nonlinear system identification. The nonlinear regression in \cite{ziemann2022learning,ziemann2023tutorial,ziemann2024sharp} is based on non-parametric LSE and its variants, and their convergence rates under different scenarios have been analyzed. It is interesting to note that this line of work usually assumes certain persistent excitation assumptions,\footnote{For example, \cite{ziemann2022learning} assumes hyper-contractivity, and \cite{ziemann2024sharp} assumes the empirical covariance of the $\{x_t\}_{t\geq 0}$ process is invertible with high probability (Corollary 3.2).} whereas our paper demonstrates that persistent excitation holds by establishing the BMSB condition for linearly parameterized and real-analytic nonlinear control systems.

Uncertainty set estimation  is crucial for robust control under model uncertainties \cite{lu2023robust,lorenzen2019robust,li2021online,akccay2004size,bai1998convergence}. SME is a widely adopted uncertainty set estimation method in robust adaptive control \cite{lorenzen2019robust,lu2023robust,bertsekas1971control,bai1995membership}. Recently, there is an emerging interest in bridging SME with statistical learning \cite{li2024icml,lu2019robust,xu2025sample} and applying SME to modern applications\cite{yu2023online,tang2024uncertainty,xu2026stochastic}. Though SME  does not require stochastic properties to be implemented, recent statistical-learning inspired analysis on SME usually consider stochastic settings \cite{akccay2004size,li2024icml,lu2023robust}.

 A \textbf{preliminary version} of this work has appeared as \cite{musavi2024identification}.
The current manuscript  extends the conference paper by
(i) providing more discussions and theoretical analysis for the system identification under randomly perturbed closed-loop policies; 
(ii) introducing counterexamples to discuss  the assumptions/conditions for our non-asymptotic analysis;
 (iii) providing additional numerical experiments to compare LSE and SME with other methods in the literature; (iv) providing a new version of the proof for Theorem \ref{th:BMSB-o-l}.

\vspace{4pt}

\noindent\textbf{Outline.} Section~\ref{sec:prob} presents the problem formulation and discusses the assumptions needed for our analysis. Section~\ref{sec:results} analyzes the non-asymptotic convergence rates of LSE.  Section \ref{sec: main text proofs} provides  proofs for Section \ref{sec:results}. Section \ref{sec: SME} reviews SME and analyzes its   convergence rates. Section~\ref{sec:counter example} provides counter-examples to demonstrate the importance of our key assumptions in Section~\ref{sec:prob}. Section \ref{sec:num-exp} provides the numerical results. Appendix provides additional proofs.

\vspace{4pt}

\vspace{4pt}

\noindent\textbf{Notations.}  For matrix $Z$, $Z^{\intercal}$ represents its transpose, $\|Z\|_{p}$ represents its $p$-norm, where $p=1, 2, \infty$, and $\|Z\|_F$ represents the Frobenius norm. Let $I_n$ denote the $n \times n$ identity matrix.
For a real symmetric matrix $Z$, $Z \succ 0$ and $Z \succeq 0$ indicate that $Z$ is positive definite and positive semi-definite, respectively.  For a measurable set $\mathcal{E} \subset \mathbb{R}^{n}$, $\texttt{Leb}^{n}(\mathcal{E})$ represents its Lebesgue measure in $\mathbb{R}^{n}$. 
For a set  of matrices in $\mathcal{T} \subseteq \mathbb{R}^{n \times m}$, the diameter is defined as $\hbox{diam}(\mathcal{T}) = \sup_{A_1, A_2\in \mathcal{T}} \|A_1-A_2\|_{F}$.
For random variables \(Z_1,\ldots,Z_n\), we use
$\mathcal{F}(Z_1,\ldots,Z_n)$
to denote the \(\sigma\)-algebra generated by \(Z_1,\ldots,Z_n\). A function $\gamma(\cdot)$ is a
$\mathcal{K}$ function if it is continuous, strictly increasing, and satisfies
$\gamma(0) = 0$. A function $\beta(\cdot, \cdot)$ is a $\mathcal{KL}$ function if, for each fixed $t \geq 0$,
the function $\beta(\cdot, t)$ is a $\mathcal{K}$ function, and for each fixed $s \geq 0$, the function $\beta(s,\cdot)$ is decreasing and satisfies $\beta(s,t) \rightarrow 0$ as $t \rightarrow \infty$.
For a function $g$ defined on a set $\mathcal{D}$, we write $g \equiv 0$ if
$g(\zeta) = 0$ for all $\zeta \in \mathcal{D}$, and $g \not\equiv 0$ otherwise. 
The  functions $\phi_1, \dots, \phi_{n_\phi}$ are called
{linearly independent} if $\sum_{i=1}^{n_\phi} c_i \phi_i \not\equiv 0$
for every non-zero vector  $c=(c_1, \dots, c_{n_{\phi}})$. 
The notation $\lceil \cdot \rceil$ stands for the ceiling function.  $\sigma_{\min}(Z)$ denote the  minimum singular value of matrix $Z$. The notation $\emptyset$ stands for an empty set. The complement of a set $\mathcal Z$ is denoted as $\Z^{\mathrm c}$. We use $\texttt{trunc-G}(0, \sigma_{w}, \mathcal D)$ to denote the
truncated Gaussian distribution with zero mean and $\sigma_{w}^{2}$ variance truncated to $\mathcal D$. $\texttt{uniform}(\mathcal D)$ denotes the uniform distribution on $\mathcal D$.

%% file: Problem_Formulation.tex
\section{Problem Formulation and Preliminaries}\label{sec:prob}

This paper studies the identification of linearly parameterized nonlinear (LPN) systems as follows:
\begin{equation}\label{eq:sys}
    x_{t+1} = \Theta_{*}\phi(x_{t},u_{t}) + w_{t},
\end{equation}

where $x_{t}\in\mathbb{R}^{n_{x}}$ is the state, $u_{t}\in\mathbb{R}^{n_{u}}$ is the control input, an equilibrium point is reached at $(0,0)$, i.e., $0=\Theta_*\phi(0, 0)$, $w_{t}\in\mathbb{R}^{n_{x}}$ is i.i.d. random  process noise,  $\Theta _*\in \R^{n_x \times n_{\phi}} $ is an unknown parameter matrix, and  $\phi: \R^{n_{x}+n_{u}}\rightarrow \R^{n_{\phi}}$ is a known nonlinear feature mapping whose components $\phi(\cdot)=(\phi_{1}(\cdot), \ldots, \phi_{n_{\phi}}(\cdot))^{\intercal}$ are linearly independent.\footnote{If feature functions are not independent, they can be reduced to an independent set since they are known a priori.} For simplicity, we consider $x_0=0$ throughout this paper.\footnote{{This  could be generalized to any non-zero $x_0$  that satisfies the initial condition of the local input-to-state stability as introduced in Definition \ref{def:liss}.}}

In this paper, we aim  to estimate the unknown parameters $\Theta_*$ from a single trajectory of samples: $\{x_t, u_t, x_{t+1}\}_{t=0}^{T-1}$. In particular, we consider the least-squares estimator (LSE) below:
\begin{align}\label{eq:ols-eq}
    \hat{\Theta}_{T} 
    = \underset{\hat{\Theta}}{\arg\min} 
    \sum_{t=0}^{T-1} \big\|x_{t+1} - \hat{\Theta}\phi(x_{t}, u_{t}) \big\|_{2}^{2},
\end{align}
where $ \hat{\Theta}_{T}$ represents the estimated parameter matrix after a trajectory of $T$ samples. 

This paper will focus on non-active exploration. In particular, we will first consider i.i.d. random control inputs, i.e., 
\begin{equation}\label{equ: iid input}
    u_t=\eta_t
\end{equation}for some i.i.d. random variable $\eta_t$. Next, in Section \ref{subsec: randomly perturbed policy LSE}, we will extend the results to randomly perturbed control policies:
\begin{equation}\label{equ: randomly perturbed control}
    u_t= \pi(x_t)+\eta_t
\end{equation}
where $\pi(\cdot)$ can be any real analytic policy  and is usually  task/performance-driven, and  $\eta_t$ is an i.i.d. random variable that aims to explore the system and generate informative samples. Such extension has been shown to preserve the sample complexity of LSE under i.i.d. inputs for linear system identification \cite{li2023non}. 

Due to its simplicity, non-active exploration is widely adopted in linear system identification and enjoys strong theoretical guarantees \cite{simchowitz2018learning,li2023non,sattar2022finite}. For example,  LSE with non-active  exploration can achieve the optimal convergence rate for linear systems \cite{simchowitz2018learning}. Later, similar sample complexity bounds have been developed for some sub-classes of nonlinear systems,  such as bilinear systems \cite{sattar2022finite}  and linear systems with nonlinear policies \cite{li2023non}. The primary goal of this paper is to study/establish the conditions of  nonlinear systems \eqref{eq:sys} that enable  system identification with provable finite-sample guarantees under non-active explorations. 

\subsection*{Conditions and Assumptions}



In the following, we  introduce two conditions that together enable efficient system identification under non-active explorations:  one on the feature function and one on the noise distribution. Further, we  introduce an  assumption on stability for the purpose of analysis, which is needed in most nonlinear system identification literature \cite{sattar2025finite,mania2022active,sattar2022finite}.

\vspace{4pt}

\noindent\textbf{Condition on feature functions.}  First, we introduce the following regularity condition on the feature functions. 

\vspace{2pt}

\begin{definition}[Real Analyticity \cite{krantz2002primer}]
A function $f(\cdot)$ is said to be real-analytic if, for every $x$, the function $f$ may be represented by a convergent power series in a neighborhood of $x$.  

\end{definition} 

\vspace{2pt}

It is known that a real-analytic function is infinitely differentiable and the convergent power series is unique and coincides with the infinite Taylor expansion, but infinite differentiability does not imply real analyticity \cite{krantz2002primer}. 

This paper mainly focuses on real-analytic feature functions.
\begin{assumption}[Real-Analytic Features]\label{ass:analytic}
Every feature function $\phi^i(x,u)$ for $1\leq i \leq n_{\phi}$ is  real-analytic on $\mathbb{R}^{n_x+n_u}$.\footnote{This assumption can be relaxed to local analyticity on the reachability set.}
\end{assumption}

Real analytic functions include many common function types, such as polynomial functions, trigonometric functions, exponential functions, etc. Further, real analyticity is preserved under addition and multiplication (Proposition 1.1.7 in \cite{krantz2002primer}). Therefore, real analytic feature functions naturally arise in  many applications\cite{simpson2016voltage,noack2003hierarchy,siciliano2010robotics,alaimo2013mathematical}, some of which are briefly discussed below for illustration purposes.

\begin{example}[Pendulum]\label{ex:pendulum}
Consider a discrete-time pendulum system below:
\begin{align*}
\alpha_{t+1} &= \alpha_{t} + \Delta_T \omega_t+w_{t,1}, \\
\omega_{t+1} &= \omega_{t}-\frac{g\Delta_T}{l}\sin(\alpha_{t})
+ \frac{\Delta_T}{M l^2}u_t
+ \Delta_T w_{t,2},
\end{align*}
where $\alpha_t$ represents the angle, $\omega_t$ represents the angular velocity, $u_t$ is the torque input, $g$ is the gravity constant, $\Delta_T$ is the time discretization,   the mass $M$ and rod length $l$ are unknown. This system could be viewed as an LPN system \eqref{eq:sys}, and the feature functions are real-analytic because they are either linear functions or a summation of linear terms and  a $\sin(\alpha_t)$ term, which are all real analytic.

\end{example}

\begin{example}[Drone]\label{ex:quadrotor}
Consider the drone dynamics in \cite{alaimo2013mathematical} as reviewed below:
\begin{align*}
    \dfrac{d}{dt} \begin{pmatrix}
        p\\
        v\\
        q\\
        \omega
    \end{pmatrix} = \begin{pmatrix}
        v\\
        - g e_{z} + \frac{1}{M} Q(q) f_{u}\\
        \frac{1}{2} \Omega(\omega) q\\
        I^{-1} (\tau_{u} - \omega \times I \omega) 
    \end{pmatrix} + w,
\end{align*}
where $p \in \mathbb{R}^{3}$ and $v \in \mathbb{R}^{3}$ represent the  position and velocity  in the inertial frame,  $\omega \in \mathbb{R}^{3}$ denotes the angular velocity in the body-fixed frame,  $q \in \mathbb{R}^{4}$ denotes the quaternion vector \cite{alaimo2013mathematical}, $f_u$ and $\tau_u$ are control inputs, the mass $M$ is unknown,  the diagonal inertia  matrix $I$ is unknown,  $Q(q)$ is a known quadratic function, and $\Omega(\omega)$ is a known linear function. After time discretization, this system can be written as an LPN system \eqref{eq:sys}, and the feature functions are third-order polynomial functions and are thus real-analytic. 
\end{example}

\begin{example}[Neural Network]\label{ex:neural network}
 The Gaussian Error Linear Unit (GELU) is a widely used activation function in modern neural network architectures, including BERT, GPT, and Vision Transformers \cite{hendrycks2016gaussian,brown2020language,devlin2019bert,dosovitskiy2021vit}. GELU function is defined by
 \begin{equation*}
\textup{GELU}(x)
=
\frac{x}{\sqrt{2\pi}}
\int_{-\infty}^{x}
e^{-t^{2}/2}\,dt.
\end{equation*}
 GELU function is real-analytic because the exponential function $e^{-t^{2}/2}$ is real-analytic, and the real analyticity is preserved under integration (Proposition 1.1.18 in \cite{krantz2002primer}) and multiplication with another real analytic function $x$. Therefore, when GELU is  used to approximate the nonlinear dynamics while only the weights of the final linear layer are adapted, the resulting model is an LPN system and the feature functions  generated by  GELU activation are real analytic.
\end{example}

\vspace{4pt}

\noindent\textbf{Condition on noise distributions.} Next, we introduce the assumption on the randomness of the system \eqref{eq:sys}, which relies on the following definition.

\begin{definition}[Semi-Continuous Distribution \cite{pollard2002user}]\label{def:semi-continuous}
A probability distribution $\mathbb{P}$ on $\mathbb R^n$ is called {semi-continuous}   if no set with Lebesgue measure zero has full probability measure, i.e., if  set $\mathcal E$ satisfies $\texttt{Leb}^{n}(\mathcal E) = 0 $, then $\Pb(\mathcal E)<1$. 
\end{definition}

The semi-continuous distribution is more general than the continuous distribution but rules out purely discrete distributions \cite{cohn2013measure,pollard2002user}. In particular, any continuous distribution satisfies the semi-continuity condition in Definition \ref{def:semi-continuous}, and any mixture distribution with one component being a continuous distribution  also satisfies the requirement of semi-continuity.

Based on the semi-continuity definition, we introduce the following assumptions on $w_t$ and $\eta_t$. 
\begin{assumption}[Noise Distributions]\label{ass:bounded-iid-noise}
The process noise $\{w_t\}_{t\ge 0}$ and the control perturbations $\{\eta_t\}_{t\ge 0}$ are mutually independent and both satisfy the following properties
\begin{enumerate}
\item[(i)] Both sequences are i.i.d.  sub-Gaussian with  zero means, positive definite covariance matrices, i.e., $\operatorname{cov}\left(w_t\right) \succeq \sigma_w^2 I_{n_x} \succ 0$, $\operatorname{cov}\left(\eta_t\right) \succeq \sigma_\eta^2 I_{n_u} \succ 0$, and sub-Gaussian  proxies $\kappa_w$, $\kappa_\eta$, respectively.
    \item[(ii)] The noises are bounded: $\|w_t\|_\infty \le w_{\max}$ and
    $\|\eta_t\|_\infty \le \eta_{\max}$ almost surely.
    \item[(iii)]  The distributions of $w_t$ and $\eta_t$ are semi-continuous.
\end{enumerate}
\end{assumption}

\vspace{3pt}

In Assumption \ref{ass:bounded-iid-noise}, the  properties in (i) are standard in the literature of  finite-sample analysis of system identification  \cite{simchowitz2018learning,simchowitz2020naive,sattar2022finite}. The  property (ii) assumes  boundedness on  process noises and control inputs. Though boundedness is stronger than the sub-Gaussian assumption in the linear system identification literature, it is a common assumption for nonlinear system identification \cite{mania2022active,shi2021meta,li2023online}.  Further, in many physical applications, noises are usually bounded, e.g. the wind disturbances in drone systems are bounded, the renewable energy injections in power systems are also bounded, etc. The property (iii) on semi-continuity  may seem restrictive, since it rules out the discrete distributions. However, the disturbances in many realistic systems can satisfy the semi-continuity because realistic noises are usually generated from a mixture distribution where at least one component is continuous. Finally, here we impose similar conditions on $w_t$ and $\eta_t$ with possibly different parameters, which is a common practice in the literature of finite sample analysis of system ID \cite{li2023non,simchowitz2018learning,sattar2022finite}.

\vspace{4pt}

\noindent\textbf{Assumption on stability.} Finally, to ensure well-posed system behavior, we assume   locally input-to-state stable systems, whose definition is reviewed below. 

\vspace{2pt}

\begin{definition}[Local Input-to-State Stability]\label{def:liss}
Consider a general nonlinear system $x_{t+1} = f(x_t, d_t)$ with  $x_t \in \mathbb{R}^{n_x}$,  $d_t \in \mathbb{R}^{n_d}$, and
$f(0,0) = 0$. The  system is called  \emph{locally input-to-state stable (LISS)} 
if there exist parameters $\rho_x,\rho_d>0$,  and functions $\beta(\cdot, \cdot) \in \mathcal{KL}$ and $
\gamma(\cdot) \in \mathcal{K}$ such that, for every initial state
$\|x_0\|_2 \le \rho_x$ and every input sequence 
$\sup_{t \ge 0}\|d_t\|_\infty \le \rho_d$, we have
\[
  \|x_t\|_2 \;\le\; \beta\big(\|x_0\|_2,\, t\big)
  + \gamma\Big(\sup_{s \ge 0}\|d_s\|_\infty\Big), \qquad \forall\, t \ge 0 .
\]

\end{definition}

\vspace{6pt}

For simplicity, we first assume that our system \eqref{eq:sys} is open-loop stable in Assumption \ref{ass:liss} below. In Section \ref{subsec: randomly perturbed policy LSE}, we will relax this assumption to closed-loop stability with a known stabilizing controller $\pi$ as in \eqref{equ: randomly perturbed control}. The analysis and theoretical results are very similar.
\begin{assumption}[Open-Loop Stability]\label{ass:liss}
System \eqref{eq:sys} is LISS under control inputs $u_t =\eta_t$  with $d_t=(\eta_t, w_t)$ and  parameter  $\rho_d\geq \max(\eta_{\max}, w_{\max})$.

\end{assumption}

We remark that stability is usually assumed in the finite sample analysis of system identification of nonlinear systems \cite{foster2020learning,sattar2022non,li2023online,lin2024online,sattar2025finite}. In this paper, we adopt the notion of local input-to-state stability (ISS), which is a standard stability notion for nonlinear systems and is less restrictive than global ISS. Together with the bounded noises imposed by Assumption \ref{ass:bounded-iid-noise}, Assumption \ref{ass:liss} guarantees bounded state and action trajectories, which are crucial for our theoretical analysis. It is worth mentioning that these bounded states and actions are commonly assumed in general learning-based nonlinear system identification and control literature \cite{sattar2022non,foster2020learning,li2023online}.

Lastly, we comment on the tradeoff/relation between the assumptions on system stability and disturbance. If disturbances are unbounded, then certain global stability properties are needed, as  in the literature \cite{foster2020learning,sattar2022non,li2023online,lin2024online}. By contrast, for bounded disturbances, local stability can be sufficient as long as the stability region covers the reachability set. Since real-world systems typically encounter bounded disturbances/inputs and generally only satisfy local stability~\cite{slotine1991applied}, we address this trade-off through our current combination of assumptions. It is left for future  to study other combination of assumptions. 

%% file: Main_Results.tex
\section{Theoretical Results}\label{sec:results}
This section provides the convergence rates of least square estimation \eqref{eq:ols-eq} for linearly parameterized nonlinear systems \eqref{eq:sys}. In the following, we  first develop the block-martingale small-ball (BMSB) condition for our LPN systems, which is a crucial condition for the finite sample analysis of  system identification for both linear  systems  and LPN systems. Based on the BMSB condition, we  formally establish  the non-asymptotic convergence rate of LSE under i.i.d. random inputs \eqref{equ: iid input} in Section \ref{subsec: random input LSE}. Lastly, Section \ref{subsec: randomly perturbed policy LSE} generalizes the  results to randomly perturbed control policies. 

\subsection{Block-Martingale Small-Ball Condition}
The convergence rate analysis throughout this paper relies on a crucial condition introduced in \cite{simchowitz2018learning}: block-martingale small-ball (BMSB) condition, which is  reviewed below. 

\vspace{4pt}

\begin{definition}[Block-Martingale Small-Ball Condition~\cite{simchowitz2018learning}]\label{def:bmsb}
Let $\{\mathcal{F}_{t}\}_{t \geq 0}$ denote a filtration and let $\{y_{t}\}_{t \geq 0}$ be an $\{\mathcal{F}_{t}\}_{t \geq 0}$-adapted random process taking values in $\mathbb{R}^{n_{y}}$. We say $\{y_{t}\}_{t \geq 0}$ satisfies the $(k, \Gamma_{sb}, p)$-block martingale small-ball (BMSB) condition for a positive integer $k$, a  matrix $\Gamma_{sb} \succ 0$, and a scalar $p \in (0, 1)$ if the process $\{y_{t}\}_{t \geq 0}$ satisfies
\begin{align*}
    \frac{1}{k} \sum_{i=1}^{k} \mathbb{P}\big( |v^{\intercal} y_{t+i}| \geq \sqrt{v^{\intercal} \Gamma_{sb} v}\ \big|\ \mathcal{F}_{t}\big) \geq p
\end{align*}
almost surely for any $t \geq 0$   and for any  $v \in \mathbb{R}^{n_{y}}$ such that $\|v\|_{2} = 1$.
\end{definition}

\vspace{4pt}

The BMSB condition is a key condition in the finite-sample analysis of system ID  for both linear systems \cite{simchowitz2018learning} and nonlinear systems \cite{sattar2022finite}. Roughly speaking, BMSB requires that, over any time block of length $k$,  the trajectory $\{y_t\}_{t\geq 0}$ cannot remain concentrated in a $(n_y-1)$-dimension subspace orthogonal to a unit direction vector $v$ for any direction $v$ from a unit sphere. This directional anti-concentration condition provides sufficient exploration across all directions as the total horizon $T$ grows, enabling the success of system identification. Further, the BMSB condition has a close relation with  the classic persistent excitation condition in the  control literature \cite{ljung1998system}. In particular,  BMSB has been shown to imply persistent excitation  with high probability (see Lemma 4.1 in \cite{li2024icml}), so it can be viewed as a stochastic/probabilistic version of persistent excitation.

One of our major technical contributions  is formally establishing the BMSB condition for LPN systems, which is introduced  in the theorem below.

\begin{theorem}[BMSB]\label{th:BMSB-o-l}
Under Assumptions~\ref{ass:analytic},~\ref{ass:bounded-iid-noise},~\ref{ass:liss} and i.i.d. random inputs \eqref{equ: iid input}, there exist positive constants $s_{\phi}>0$ and $0<p_{\phi} <1$ such that the feature vectors $\{\phi(x_{t}, u_{t})\}_{t \geq 0}$ satisfy the $(1,\, s_{\phi}^{2}I_{n_{\phi}},\, p_{\phi})$-BMSB condition, where the corresponding filtration is $\mathcal{F}_{t} = \mathcal{F}(w_{0}, \dots, w_{t-1}, x_{0}, \dots, x_{t}, \eta_{0}, \dots, \eta_{t})$.

\end{theorem}

\vspace{2pt}

The proof of Theorem \ref{th:BMSB-o-l} is provided in Section \ref{subsec: proof of Thm 1}. 

Next, we provide discussions on Theorem \ref{th:BMSB-o-l}. Firstly, it is worth mentioning that BMSB is a crucial enabling condition for the recent finite-sample analysis of system identification. Therefore, Theorem \ref{th:BMSB-o-l} lays foundation for the convergence rate analysis for a broad range of system identification methods. Though this section  mainly studies the least square estimation,  we will show that another system identification method, set-membership estimation, can also be analyzed for its convergence rate under the stochastic setting, benefiting from the BMSB condition.

Secondly, Theorem \ref{th:BMSB-o-l} assumes open loop stability and establishes BMSB under i.i.d. random inputs \eqref{equ: iid input}. In the next subsection, we will relax it to closed-loop stability and establish BMSB under randomly perturbed closed-loop policies \eqref{equ: randomly perturbed control}. The analysis and results are similar, but the applicability is greatly generalized.

Thirdly, notice that our theorem uses block length $k=1$, so essentially, we have proved that our system satisfies the marginal small ball condition due to block length being 1. This is a common effect in the system identification literature, for example, linear systems and bilinear systems all satisfy BMSB with $k=1$. 

Lastly, Theorem \ref{th:BMSB-o-l} establishes the existence of proper BMSB parameters $s_{\phi}, p_{\phi}$ that are independent of  horizon $T$. This is sufficient for establishing the $O(\frac{1}{\sqrt T})$ convergence rate in Theorem \ref{thm: LSE open loop}.  The explicit formulas of $s_{\phi}, p_{\phi}$ and their dependence on dimensionality are left as future work. It is particularly challenging to establish  generic formulas for all LPN systems, so an interesting future direction is to study reasonable sub-classes of nonlinear systems and discuss the dependence of constants $(s_{\phi}, p_{\phi})$ on the system characteristics.

\subsection{Convergence Rates of Least Square Estimation}\label{subsec: random input LSE}
In this subsection, we introduce and discuss the non-asymptotic convergence rate of least square estimation \eqref{eq:ols-eq} when learning the unknown parameters in LPN systems \eqref{eq:sys}.

\begin{theorem}[Convergence Rate of LSE]\label{thm: LSE open loop}
Consider Assumptions~\ref{ass:analytic},~\ref{ass:bounded-iid-noise},~\ref{ass:liss} and i.i.d. random inputs \eqref{equ: iid input}. {There exists $\phi_{\max}<+\infty$ such that $\|\phi(x_t,u_t)\|_2 \leq \phi_{\max}$ for all $t$}. Further,
for any $\delta \in (0, 1/3)$, when the horizon $T$ is large enough, i.e.,
\begin{equation*}
\begin{aligned}
    T \;\ge\; \frac{10}{p_{\phi}^2}\bigg(
    &\log\big(1/\delta\big)
    + 2n_{\phi}\log\bigg(\frac{10\phi_{\max}}{p_{\phi}s_{\phi}}\bigg)\bigg),
\end{aligned}
\end{equation*}
with probability at least $1-3\delta$, the least square estimator \eqref{eq:ols-eq} enjoys the following non-asymptotic estimation error bound:
\begin{equation*}
\begin{aligned}
    &\big \|\hat{\Theta}_{T} - \Theta_{*} \big\|_{2}
    \leq\frac{90 \kappa_w}{p_{\phi}} \sqrt{ \frac{n_{x} +\log\Big(\frac{1}{\delta}\Big) + n_{\phi}\log\Big(\frac{10 \phi_{\max}^2}{p_{\phi}s_{\phi}^{2}}\Big)}{T s_{\phi}^{2}}},
\end{aligned}
\end{equation*}
where $s_{\phi}$ and $p_{\phi}$ are defined  in Theorem~\ref{th:BMSB-o-l}.

\end{theorem}

The proof is provided in Section \ref{subsec: LSE convergence}. 

We provide the following discussions below.
Firstly, Theorem \ref{thm: LSE open loop} shows that least square estimator $\hat \Theta_T$ converges to the true parameter matrix $\Theta_*$  at a rate of $\frac{1}{\sqrt T}$ when learning LPN system parameters \eqref{eq:sys} under i.i.d. random inputs \ref{equ: iid input}. This convergence rate is  consistent with the convergence rates of LSE for linear systems \cite{simchowitz2018learning,simchowitz2020naive} and bilinear systems \cite{sattar2022finite} under  i.i.d. random inputs. Therefore, we have shown that, for a broader range of linearly parameterized nonlinear systems, LSE with non-active exploration still enjoys the same convergence rate in terms of $T$ when considering real-analytic feature functions and semi-continuous noise distributions.

Secondly, the convergence rate relies on a large enough $T$, and the lower bound for $T$ in Theorem \ref{thm: LSE open loop} is a constant that does not increase with $T$. Therefore, as $T$ increases, the condition on $T$ can be naturally satisfied.

Thirdly,  our convergence rate relies on  uniformly bounded feature vectors for all $t$, i.e., $\|\phi(x_t,u_t)\|_2 \leq \phi_{\max}$. This is enabled by the bounded noise assumed in Assumption \ref{ass:bounded-iid-noise} and the local input-to-state stability assumed in Assumption \ref{ass:liss}. The boundedness of feature vectors is commonly assumed in the nonlinear system identification literature \cite{mania2022active,tupe2026federated,wang2025logarithmic}.

Fourthly, regarding the dimension dependence in the convergence rate,  the explicit dependence is  $\sqrt{n_x+n_\phi}$, where $n_x$ and $n_\phi$ refer to the dimensions of the state and the feature vector, respectively. But it is worth mentioning that the parameters $s_{\phi}, p_{\phi}$ may implicitly depend on the dimensions as well. For some special systems, such as bilinear systems, it has been shown that these constants are independent of the dimensions~\cite{sattar2022finite}. It is left as future work to explore other nonlinear systems’ implicit dimension dependence.

Finally, Section \ref{sec:num-exp} will discuss numerical methods to calculate these parameters and  present numerical comparisons of our theoretical bound in Theorem \ref{thm: LSE open loop} and the empirical estimation errors for different systems and noise distributions.

\subsection{Randomly Perturbed Control Policies}\label{subsec: randomly perturbed policy LSE}

So far, we only study the case where the LPN system \eqref{eq:sys} is open-loop stable, as  stated in Assumption \ref{ass:liss}. However, in practice, most nonlinear systems are not stable and have to be stabilized by certain well-designed controllers.  Therefore, this subsection generalizes the results in the previous subsections to nonlinear systems with a known stabilizing controller.

First, we state the following assumption on a  control policy that guarantees closed-loop stability. 
\begin{assumption}[Stabilizing Control]\label{ass: stabilizing controller}
    There is a real-analytic control policy  $\pi(\cdot)$  such that the closed-loop system with plant dynamics \eqref{eq:sys} and control policy \eqref{equ: randomly perturbed control} is LISS with $d_t=(\eta_t, w_t)$ and  parameter  $\rho_d\geq \max(\eta_{\max}, w_{\max})$.
\end{assumption}

\vspace{2pt}

Compared with Assumption \ref{ass:liss}, Assumption \ref{ass: stabilizing controller} provides flexibility to choose any controller $\pi$ that stabilizes the closed-loop system. If the policy $\pi\equiv 0$, Assumption \ref{ass: stabilizing controller} becomes identical to Assumption \ref{ass:liss}. 
Determining a stabilizing policy $\pi$ for uncertain  systems has been studied in robust control \cite{dullerud2013course,lu2019robust}. Relaxation from open-loop stability to closed-loop stability with a known stabilizing policy is common in system identification literature \cite{li2023non,simchowitz2018learning}.

Based on closed-loop stability, we establish the BMSB condition and the convergence rate of LSE. 
\begin{corollary}[Randomly Perturbed Control Policy]\label{cor: bmsb closed loop} Under Assumptions~\ref{ass:analytic},~\ref{ass:bounded-iid-noise}, considering a randomly perturbed control policy \eqref{equ: randomly perturbed control} that satisfy Assumption \ref{ass: stabilizing controller},  the feature vectors $\{\phi(x_{t}, u_{t})\}_{t \geq 0}$ satisfy   $(1,\, \tilde s_{\phi}^{2}I_{n_{\phi}},\, \tilde p_{\phi})$-BMSB condition with $\tilde s_{\phi}>0$ and $0<\tilde p_{\phi} <1$. Further, there exists $\tilde \phi_{\max}<+\infty$ such that $\|\phi(x_t,u_t)\|_2 \leq \tilde \phi_{\max}$ for all $t$. Consequently, LSE  enjoys the convergence rate  $O(\frac{1}{\sqrt T})$.
\end{corollary}

\vspace{3pt}

The proof is  similar to the proofs of Theorem \ref{th:BMSB-o-l} and Theorem \ref{thm: LSE open loop} and deferred to Appendix~\ref{appendix: LSE closed loop cor}.

%% file: proof_main.tex
\section{Main Proofs}\label{sec: main text proofs}
This section provides the proof of our main technical results: BMSB condition and LSE's convergence rate (Theorems \ref{th:BMSB-o-l},  \ref{thm: LSE open loop}).

\subsection{Proof of Theorem \ref{th:BMSB-o-l}}\label{subsec: proof of Thm 1}
The proof utilizes the real-analyticity of feature functions and the semicontinuity of the noise distributions to establish the BMSB condition. We will present the proof in five steps.

\vspace{2pt}

\noindent\textbf{Step 1: Preparation.} This step introduces useful shorthand notations and supporting lemmas. 

First, we introduce a convenient notation $z_t=(x_t, u_t)$ to represent the state-action pair for the rest of this paper and let $n_z\coloneqq n_x+n_u$.  $z_t$ is bounded by a compact set $\mathcal Z$ as shown in the  supporting lemma below.

\begin{lemma}\label{lem: compact Z}
    Under Assumptions \ref{ass:bounded-iid-noise}, \ref{ass:liss}, and i.i.d. random inputs \eqref{equ: iid input}, there exists a compact set $\mathcal Z$ such that $z_t\in \mathcal Z$ for all $t\geq 0$.
\end{lemma}
\begin{proof}
Let's first establish the bound on $u_t$. 
    Since $u_t=\eta_t$ and $\eta_t$ is bounded according to Assumption \ref{ass:bounded-iid-noise}, the control inputs are all bounded by a compact set $\mathcal U\coloneqq \{u: \|u\|_\infty\leq \eta_{\max}\}$. 

Next, we establish the bound on $x_t$. By Assumption \ref{ass:liss}, 
\[
  \|x_t\|_2 \le \beta\big(\|x_0\|_2,\, t\big)
  + \gamma\Big(\sup_{s \ge 0}\|d_s\|_\infty\Big), \qquad \forall\, t \ge 0,
\]

\vspace{-5pt}
\noindent{where} $d_t=(\eta_t, w_t)$. By Assumption \ref{ass:bounded-iid-noise}, we have $\|d_t\|_\infty \leq \max(\|\eta_t\|_\infty, \|w_t\|_\infty)\leq \max(\eta_{\max}, w_{\max})$ for all $t$. Further, recall that we consider $x_0=0$ for simplicity and $\beta(0, t)=0$ by the definition of $\mathcal{KL}$ functions. Therefore, we have
\begin{align*}
    \|x_t\|_2 \le \gamma\Big(\max(\eta_{\max}, w_{\max})\Big), \quad \forall \, t\geq 0.
\end{align*}

\vspace{-4pt}
\noindent{Define} $\mathcal X\coloneqq \{x: \|x\|_2 \leq \gamma\Big(\max(\eta_{\max}, w_{\max})\Big) \}$. Set $\mathcal X$ is compact and  $x_t \in \mathcal X$ for all $t$.

Let $\mathcal Z \coloneqq\mathcal X \times \mathcal U $. 
Notice that $\mathcal Z$ is compact because  $\mathcal X$ and $\mathcal U$ are compact, and $z_t =(x_t, u_t) \in \mathcal X \times \mathcal U= \mathcal Z$ for all $t$. \end{proof}

Next, we can prove the boundedness of the feature vector.
\begin{lemma}\label{lem: bdd phi(z) on Z}
  There exists $\phi_{\max} <+\infty$  such that $ \|\phi(z)\|_2\leq \phi_{\max}$ for any $z\in \Z$.
\end{lemma}
\begin{proof}
    The proof is straightforward:  function $\|\phi(z)\|_2$ is continuous, $\Z$ is compact, and a continuous function  on a compact set is always bounded.
\end{proof}

In addition, we introduce a convenient notation for the set of  noises $(\eta_t, w_t)$:
$$\mathcal D=\{\eta: \|\eta\|_\infty \leq \eta_{\max}\}\times \{w: \|w\|_\infty \leq w_{\max}\}.$$
Notice that $\mathcal D$ is compact and $(\eta_t, w_t)\in \mathcal D$ for all $t$.

Lastly, we  let $\mathbb S\coloneqq \{v \in \R^{n_{\phi}}: \|v\|_2=1\}$ to denote the set of unit direction vectors $v$ in the BMSB condition. $\Sb$ is also compact by definition.

\vspace{2pt}

\noindent\textbf{Step 2: Anti-Concentration Inequality.} To establish the $(1,\, s_{\phi}^{2}I_{n_{\phi}},\, p_{\phi})$-BMSB condition stated in Theorem \ref{th:BMSB-o-l}, it is sufficient to show that there exist positive constants $s_{\phi}, p_{\phi}$  such that
\begin{align}\label{eq:goal-bmsb}
    {\mathbb{P}}\bigg(|v^{\intercal}\phi(z_{t+1})| \geq s_{\phi} \ \big|\ {\mathcal{F}}_{t}\bigg) \geq p_{\phi}. 
\end{align}
for any $v\in \mathbb S$, any $t\geq 0$, and any filtration $\mathcal F_t$.

Notice that \eqref{eq:goal-bmsb} is essentially an anti-concentration property because it requires $\big|v^{\intercal}\phi(z_{t+1})\big|$ to deviate from 0 for a positive probability. To establish it, we rely on the following  anti-concentration inequality in the statistics literature. 
\begin{proposition}\label{prop: paley zygmund ineq}(Paley-Zygmund Inequality \cite{petrov2007lower}) 
Let $x$ be a non-negative random variable with finite fourth moment. For any $r \in (0, 1)$, the following holds:
\begin{align*} 
    \mathbb{P}\bigg(x > r \sqrt{\mathbb{E}[x^{2}]} \bigg) \geq (1 - r^{2})^{2} \dfrac{\mathbb{E}[x^{2}]^{2}}{\mathbb{E}[x^{4}]}. 
    \end{align*} 
\end{proposition}
    \vspace{2pt}

By leveraging Proposition \ref{prop: paley zygmund ineq}, we have
\begin{equation}
 \label{equ: bmsb paley zygmund}
\begin{aligned}
        &\mathbb{P}\Bigg( \big|v^{\intercal}\phi(z_{t+1})\big| > r \sqrt{\mathbb{E}\bigg[\big(v^{\intercal}\phi(z_{t+1})\big)^{2}\ \big|\ \mathcal{F}_{t} \bigg]}\ \ \Bigg|\ \mathcal{F}_{t} \Bigg)\\
         \geq\  &  (1 -r^{2})^{2}\dfrac{\mathbb{E}\bigg[\big(v^{\intercal}\phi(z_{t+1})\big)^{2}\ \big|\ \mathcal{F}_{t} \bigg]^{2}}{\mathbb{E}\bigg[\big(v^{\intercal}\phi(z_{t+1})\big)^{4}\ \big|\ \mathcal{F}_{t} \bigg]}.         
\end{aligned} 
\end{equation}
for any $r \in (0,1)$. The finite fourth momemnt condition in Proposition \ref{prop: paley zygmund ineq} is naturally satisfied because $\big|v^{\intercal}\phi(z_{t+1})\big| \leq \|\phi(z_{t+1})\|_2 \leq \max_{z\in \mathcal Z} \|\phi(z)\|_2<+\infty$ by Lemma \ref{lem: bdd phi(z) on Z}.

To establish the existence of positive constants $s_\phi, p_\phi$ for all $t$, $\mathcal F_t$, and $v\in \mathbb S$, the main challenge is to  prove the positiveness of the following two terms:
\begin{itemize} 
    \item $\text{Term 1}\coloneqq \underset{t\geq 0}{\inf}\ \underset{\mathcal{F}_{t}}{\inf}\ \underset{\substack{v \in \mathbb S}}{\inf}\ \mathbb{E}\big[\big(v^{\intercal}\phi(z_{t+1})\big)^{2}\ \big|\ \mathcal{F}_{t} \big] > 0$,
    \item  $\text{Term 2}\coloneqq \underset{t\geq 0}{\inf}\ \underset{\mathcal{F}_{t}}{\inf}\ \underset{\substack{v \in \mathbb S}}{\inf} \ \dfrac{\mathbb{E}\bigg[\big(v^{\intercal}\phi(z_{t+1})\big)^{2}\ \big|\ \mathcal{F}_{t} \bigg]^{2}}{\mathbb{E}\bigg[\big(v^{\intercal}\phi(z_{t+1})\big)^{4}\ \big|\ \mathcal{F}_{t} \bigg]}>0$,
\end{itemize}
which is addressed in Step 3 and Step 4, respectively.

\vspace{2pt}

\noindent\textbf{Step 3: Showing $\text{Term 1}>0$.} 
Recall that $z_{t+1}=(x_{t+1}, u_{t+1})=(\Theta_* \phi(x_t, \eta_t)+w_t, \eta_{t+1})$. Conditioning on $\mathcal{F}_{t} = \mathcal{F}(w_{0}, \dots, w_{t-1}, x_{0}, \dots, x_{t}, \eta_{0}, \dots, \eta_{t})$, we have $z_t=(x_t, \eta_t)$ being determined and $w_t, \eta_{t+1}$ being  random and independent from $\mathcal F_t$. Since $z_t\in \mathcal Z$ for all $t, \mathcal F_t$, we can lower bound $\text{Term 1}$ as follows:
\begin{align}
    \text{Term 1}& = \underset{t\geq 0}{\inf}\ \underset{\mathcal{F}_{t}}{\inf}\ \underset{\substack{v \in \mathbb S}}{\inf}\ \mathbb{E}\big[\big(v^{\intercal}\phi(\Theta_* \phi(z_t)+w_t, \eta_{t+1})\big)^{2}\ \big|\ \mathcal{F}_{t} \big]\notag\\
    & \geq \underset{z\in \Z}{\inf}\ \underset{\substack{v \in \mathbb S}}{\inf}\ \mathbb{E}\big[\big(v^{\intercal}\phi(\Theta_* \phi(z)+w, \eta)\big)^{2}\big],\label{equ: Term 1 to z in Z}
\end{align}
where the expectation in \eqref{equ: Term 1 to z in Z} is taken with respect to the independent random variables $w$ and $\eta$ following the distributions in Assumption \ref{ass:bounded-iid-noise}. 

Next, we aim to show $\mathbb{E}\big[\big(v^{\intercal}\phi(\Theta_* \phi(z)+w, \eta)\big)^{2}\big]>0$ for every $z\in \Z$ and $v\in \Sb$. 
To achieve this, we introduce the following shorthand notation for convenience:
\begin{equation}\label{equ: define hvz}
    h_{v,z}(w,\eta)= v^{\intercal}\phi(\Theta_* \phi(z)+w, \eta)
\end{equation}
Notice that $\Eb( h_{v,z}(w,\eta))^2 \geq 0$, so we only need to discuss the case when it is zero. Let's denote its zero set as $\Nvz$, i.e.,
\begin{equation}\label{equ:def of Nvz}
   \Nvz = \big\{(w, \eta) \in \D:\   h_{v,z}(w,\eta)= 0\big\}
\end{equation}

In the following, we will leverage the following properties of real-analytic functions in the literature \cite{krantz2002primer,cohn2013measure}.
\begin{proposition}[Properties of Real-Analytic Functions~\cite{krantz2002primer}]\label{prop: real analytic properties} Real-analytic functions enjoy the following properties:
\begin{itemize}
    \item If $f_1, f_2$ are real-analytic functions, $a_1 f_1+a_2 f_2$ is also real analytic for any $a_1, a_2\in \R$.
  
    \item A mapping $f: \R^n \to \R^m$ is called real-analytic if  component functions $f_1,\dots, f_m$ are real-analytic. Consider two real-analytic mappings: $f: \R^n \to \R^m$ and $g: \R^m \to \R^k$,   the composition $g\circ f: \R^n\to \R^k$ is also real-analytic. 
    \item Let $f$ be a real-analytic function on an open domain $\mathcal D \subseteq \R^n$. If $f\not\equiv 0$ on $\mathcal D$, then the zero set $\mathcal N_f=\{x \in \mathcal D: f(x)=0\}$ satisfies $\texttt{Leb}^{n}(\mathcal N_f)=0$.
\end{itemize}

\end{proposition}

\vspace{2pt}

Based on Proposition \ref{prop: real analytic properties},  the following lemma holds.

\vspace{2pt}

\begin{lemma}[Properties of $h_{v,z}$]\label{lem: real analytic hvz}
    Function $h_{v,z}(w,\eta)$ is real-analytic and $\texttt{Leb}^{n_{z}}(\Nvz) = 0$ for any $v\in \Sb, z\in \Z$. 
\end{lemma}
\vspace{2pt}

\begin{proof}
    We first show that  function $h_{v,z}(w,\eta)$ is real-analytic with respect to $(w,\eta)$ for any $v, z$. First, notice that $(w+\Theta_*\phi(z),\eta)$ is an affine mapping and thus a real-analytic mapping with respect to $(w,\eta)$ for any $v, z$. Since $\phi(\cdot)$ is real-analytic by Assumption \ref{ass:analytic}, the composition $\phi(\Theta_* \phi(z)+w, \eta)$ is real analytic by Proposition \ref{prop: real analytic properties}. Finally,  function  $h_{v,z}(w,\eta)=\sum_{i=1}^{n_{\phi}} v_i \phi_i(w+\Theta_*\phi(z),\eta)$ is a linear combination of real-analytic functions and is thus real analytic by Proposition \ref{prop: real analytic properties}.

    Next, we show  $\texttt{Leb}^{n_{z}}(\Nvz) = 0$. Define $g_v(z)\coloneqq v^{\intercal}\phi(z)$ for notational simplicity. Notice that $g_v(z)$ is real-analytic on $\R^{n_z}$ by Assumption \ref{ass:analytic} and Proposition \ref{prop: real analytic properties}. Further, recall that $\phi_1,\dots, \phi_{n_\phi}$ are linearly independent, so $g_v(z)=\sum_{i=1}^{n_{\phi}} v_i \phi_i(z)\not\equiv 0$ for  nonzero $v$, which is true for  $v\in \Sb$. Therefore, $g_v(z)$ is non-trivial for any $v\in \Sb$. Consequently, the zero set $\mathcal N_v=\{z \in \R^{n_z}: g_v(z)=0\}$ has $\texttt{Leb}^{n_{z}}(\mathcal N_v) = 0$. 
    For any $v,z$, we have $\Nvz +(\Theta_*\phi(z), 0)\subseteq \mathcal N_v$, because if $(w,\eta) \in \Nvz$, then $(w+\Theta_*\phi(z),\eta) \in \mathcal N_v$ by \eqref{equ:def of Nvz}.\footnote{We do not have $\Nvz +(\Theta_*\phi(z), 0)= \mathcal N_v$ because $\Nvz$ is defined on $\D$ but $\mathcal N_v$ is on $\R^{n_z}$. } Consequently, $0\leq \texttt{Leb}^{n_{z}}(\Nvz)= \texttt{Leb}^{n_{z}}(\Nvz +(\Theta_*\phi(z), 0)) \leq \texttt{Leb}^{n_{z}}(\mathcal N_v) =  0$. Thus, we have shown $\texttt{Leb}^{n_{z}}(\Nvz) = 0$.
\end{proof}

With Lemma \ref{lem: real analytic hvz}, we can show $\Pb(\Nvz)<1$.
\begin{lemma}\label{lem: P(Nvz)<1}
For any $v\in \Sb$ and $z\in \Z$, we have $\Pb(\Nvz)<1$. 
    
\end{lemma}
\begin{proof} If $\Pb(\Nvz)=0$, we have $\Pb(\Nvz)<1$ directly. Therefore, it is sufficient to only consider the case when $\Pb(\Nvz)>0$. 
    By Assumption \ref{ass:bounded-iid-noise}, $\eta$ and $w$ follow mutually independent and semi-continuous probability distributions. Therefore, the joint distribution of $(\eta, w)$ is also semi-continuous (see Appendix~\ref{appendix: semi continuity}). By Lemma \ref{lem: real analytic hvz}, we have $\texttt{Leb}^{n_{z}}(\Nvz) = 0$. By Definition \ref{def:semi-continuous}, we have $\Pb(\Nvz)<1$.
    \end{proof}

    \vspace{2pt}

Consequently, we have
\begin{align*}
&\mathbb{E}\big[\big(h_{v,z}(w,\eta)\big)^{2}\big] =   \mathbb{E}\big[\big(h_{v,z}(w,\eta)\big)^{2}\mid \Nvz\big] \Pb(\Nvz)\\
& \qquad \qquad \qquad\qquad + \mathbb{E}\big[\big(h_{v,z}(w,\eta)\big)^{2}\mid \Nvz^{\mathrm c}\big] \Pb(\Nvz^{\mathrm c})\\
=  &\ \mathbb{E}\big[\big(h_{v,z}(w,\eta)\big)^{2}\mid \Nvz^{\mathrm c}\big] (1-\Pb(\Nvz))>0,
\end{align*}
where the first equality is by a conditional expectation property, the second equality is by the definition \eqref{equ:def of Nvz}, and the last inequality is because $(h_{v,z}(w,\eta))^2> 0$ when $h_{v,z}(w,\eta)\not =0$ and $\Pb(\Nvz)<1$ by Lemma \ref{lem: P(Nvz)<1}. 

Therefore, we have shown $\mathbb{E}\big[\big(h_{v,z}(w,\eta)\big)^{2}\big]>0$ for every $v \in \Sb$ and $z\in \Z$. Recall that  $\Sb$ and  $\Z$ are both compact sets according to \textbf{Step 1}. Since $\mathbb{E}\big[\big(h_{v,z}(w,\eta)\big)^{2}\big]$ is a continuous  function with respect to $v,z$ by the boundedness and the Dominated Convergence Theorem \cite{pollard2002user}, the infimum over $v \in \Sb$ and $z\in \Z$ can be reached and thus
$$\underset{z\in \Z}{\inf}\ \underset{\substack{v \in \mathbb S}}{\inf}\ \mathbb{E}\big[\big(v^{\intercal}\phi(\Theta_* \phi(z)+w, \eta)\big)^{2}\big]>0.$$
Together with inequality \ref{equ: Term 1 to z in Z}, we have shown $\text{Term 1}>0$.

\vspace{2pt}

\noindent\textbf{Step 4: Showing $\text{Term 2}>0$.} Notice that
\begin{align*}
    \text{Term 2}&\geq \dfrac{\underset{t\geq 0}{\inf}\ \underset{\mathcal{F}_{t}}{\inf}\ \underset{\substack{v \in \mathbb S}}{\inf} \ \mathbb{E}\bigg[\big(v^{\intercal}\phi(z_{t+1})\big)^{2}\ \big|\ \mathcal{F}_{t} \bigg]^{2}}{\underset{t\geq 0}{\sup}\ \underset{\mathcal{F}_{t}}{\sup}\ \underset{\substack{v \in \mathbb S}}{\sup}  \ \mathbb{E}\bigg[\big(v^{\intercal}\phi(z_{t+1})\big)^{4}\ \big|\ \mathcal{F}_{t} \bigg]}\\
    & = \dfrac{(\text{Term 1})^2}{\underset{t\geq 0}{\sup}\ \underset{\mathcal{F}_{t}}{\sup}\ \underset{\substack{v \in \mathbb S}}{\sup}  \ \mathbb{E}\bigg[\big(v^{\intercal}\phi(z_{t+1})\big)^{4}\ \big|\ \mathcal{F}_{t} \bigg]}
\end{align*}
Since Step 3 has shown that $\text{Term 1}>0$, it is sufficient to show that the denominator is finite, which  is straightforward:
\begin{align*}
   & \ \underset{t\geq 0}{\sup}\ \underset{\mathcal{F}_{t}}{\sup}\ \underset{\substack{v \in \mathbb S}}{\sup}  \ \mathbb{E}\bigg[\big(v^{\intercal}\phi(z_{t+1})\big)^{4}\ \big|\ \mathcal{F}_{t} \bigg]\\
   \leq & \ \underset{z\in \Z}{\sup}\ \underset{\substack{v \in \mathbb S}}{\sup}  \ \big(v^{\intercal}\phi(z)\big)^{4} \leq \underset{z\in \Z}{\sup} \ \|\phi(z)\|_2^{4} <+\infty
\end{align*}
where the second inequality is by Cauchy-Schwarz inequality and the last inequality is by Lemma \ref{lem: bdd phi(z) on Z}.

\vspace{2pt}

\noindent\textbf{Step 5: Conclusion.} Finally, by \eqref{equ: bmsb paley zygmund} and Steps 3-4, we have
\begin{align*}
    &\Pb(\big|v^{\intercal}\phi(z_{t+1})\big| 
    \geq r\sqrt{\text{Term 1}} \mid \F_t) \\
    \geq \ & \mathbb{P}\Bigg( \big|v^{\intercal}\phi(z_{t+1})\big| > r \sqrt{\mathbb{E}\bigg[\big(v^{\intercal}\phi(z_{t+1})\big)^{2}\ \big|\ \mathcal{F}_{t} \bigg]}\ \ \Bigg|\ \mathcal{F}_{t} \Bigg)\\
     \geq \ &  (1 -r^{2})^{2}\dfrac{\mathbb{E}\bigg[\big(v^{\intercal}\phi(z_{t+1})\big)^{2}\ \big|\ \mathcal{F}_{t} \bigg]^{2}}{\mathbb{E}\bigg[\big(v^{\intercal}\phi(z_{t+1})\big)^{4}\ \big|\ \mathcal{F}_{t} \bigg]}\geq (1-r^2)^2 \cdot \text{Term 2}
\end{align*}
for all $t, \F_t, v \in \Sb$, where $0<r<1$ is a constant. Therefore, we have established the BMSB condition with positive constants $s_{\phi}=r\sqrt{\text{Term 1}}>0$ and $p_{\phi}=(1-r^2)^2 \cdot \text{Term 2}>0 $ that do not depend on horizon $T$.

\subsection{Proof of Theorem~\ref{thm: LSE open loop}}\label{subsec: LSE convergence}
First, the boundedness of feature functions is straightforward: for any $t$, $(x_t,u_t)\in \Z$ by Lemma \ref{lem: compact Z}, so $\|\phi(x_t, u_t)\|_2 \leq \phi_{\max}$  by Lemma \ref{lem: bdd phi(z) on Z}.

Next, the estimation error bound is based on the BMSB condition established in Theorem \ref{th:BMSB-o-l} and the general error bound proposed in \cite{simchowitz2018learning} for  linear regression  with correlated data. 

\begin{proposition}[LSE for General Linear Regression \cite{simchowitz2018learning}]~\label{prop: general LSE}
Consider a general linear regression $y_{t} = \Theta_{*} x_{t} + w_{t}$, where  $y_t\in \R^{n_y}$, $x_t \in \R^{n_x}$ is adapted to a natural filtration $\F_t$, and the following conditions hold:
\begin{enumerate}
    \item[(I)]$w_{t} | \mathcal{F}_{t}$ follows  a sub-Gaussian distribution with zero mean and   proxy $\kappa_w$;
    \item[(II)]  $\{x_{t}\}_{t \geq 1}$ satisfies the $(k, \Gamma_{sb}, p)$-BMSB condition;
\item[(III)]$\mathbb{P}\big(\sum_{t=1}^{T} x_{t}x_{t}^{\intercal} \not \preceq T\bar{\Gamma} \big) \leq \delta$ for a finite $\bar \Gamma$.
\end{enumerate}
For any $0<\delta <1/3$, 
when 
$$
    T \geq \frac{10k}{p^{2}} \left( \log(1/{\delta}) +  \log \det ( \bar{\Gamma} \Gamma_{sb}^{-1}) + 2n_{x}\log\left(\frac{10}{p}\right)\right),
$$
with probability at least $1-3\delta$, the LSE estimator $\hat \Theta_T=\arg\min\sum_{t=1}^T \|y_t -\hat \Theta x_t\|_2^2$ satisfies

\vspace{-5pt}
\begin{align*}
&\big \|\hat{\Theta}_{T} - \Theta_{*} \big\|_{2}\\
 \leq \ &\frac{90 \kappa_{w}}{p} \sqrt{ \frac{n_{y} \!+\!\log(\frac{1}{\delta}) +  \log \det ( \bar{\Gamma} \Gamma_{sb}^{-1}) \!+ \! n_{x}\log\bigg(\frac{10}{p}\bigg)}{T\sigma_{\min}(\Gamma_{sb})}}.
\end{align*}
 \end{proposition}

\vspace{5pt}

We complete the proof of Theorem \ref{thm: LSE open loop} by verifying the conditions assumed in Proposition \ref{prop: general LSE}. Condition I is straightforward because $w_t\mid\F_t=w_t$ and $w_t$'s distribution satisfies Condition I by Assumption \ref{ass:bounded-iid-noise}. Condition II has been verified by Theorem \ref{th:BMSB-o-l}. Condition III can be verified below. Notice that 
$$\sigma_{\max}(\phi(z_{t})\phi^{\intercal}(z_{t}))\leq \text{trace}(\phi(z_{t})\phi^{\intercal}(z_{t}))=\|\phi(z_{t})\|_2^2 \leq \phi_{\max}^2.$$
 Therefore, $\Pb\left( \sum_{t=1}^{T} \phi(z_{t})\phi^{\intercal}(z_{t}) \npreceq T\phi_{\max}^2 I_{n_\phi}\right)=0<\delta$ for any $\delta>0$. Therefore, by applying Proposition \ref{prop: general LSE}, we have completed the proof.

%% file: SME.tex
\section{Convergence Rates of Set Membership Estimation}\label{sec: SME}
Set membership estimation (SME) is an uncertainty quantification method for system identification in control literature \cite{bertsekas1971control,fogel1982value,lu2023robust,li2024icml}. Unlike LSE, SME is a set-estimator and directly estimates the uncertainty set.  Based on the BMSB condition established in Section \ref{sec:results}, we analyze the convergence rate of SME for LPN systems in this section. 

First, we review the SME algorithm below. Given data $\{x_t,u_t, x_{t+1}\}_{t=0}^{T-1}$, SME estimates an uncertainty set of the unknown parameter matrix $\Theta_*$ by
\begin{align}\label{eq:SM-eq}
    \ThetaSet = \bigcap \limits_{t=0}^{T-1} \Big\{\hat{\Theta}: x_{t+1} - \hat{\Theta} \phi(x_{t}, u_{t}) \in \mathcal{W} \Big\}, 
\end{align}
where $\mathcal{W}=\{w: \|w\|_\infty \leq w_{\max}\}$ is the bounded set such that $w_{t} \in \mathcal{W}$ for all $t \geq 0$ according to Assumption \ref{ass:bounded-iid-noise}. Notice that the true parameter matrix $\Theta_*$ is guaranteed to be included by the uncertainty set $ \ThetaSet $ as long as the disturbance set $\mathcal W$ is valid.

It is worth noting that the implementation of  SME \eqref{eq:SM-eq} only requires $w_t$ to be bounded  and does not need any stochastic properties of $w_t$. Nevertheless, the convergence and convergence rate analysis of SME usually assumes i.i.d. disturbances for analytical simplicity and  comparison with other statistical estimation methods \cite{akccay2004size,bai1998convergence,li2024icml}.\footnote{There are also convergence analysis under deterministic disturbances, e.g., \cite{livstone1996asymptotic}.} This paper also investigates SME's convergence rate under i.i.d. disturbances.

Further, the convergence analysis of SME typically requires an additional assumption: the boundary-visiting condition \cite{li2024icml,lu2019robust,akccay2004size,xu2025sample}.  
\begin{assumption}[Boundary Visiting Condition]\label{ass:w-tight}
For any $\ell > 0$, there exists $q_{w}(\ell) > 0$, such that for any $1 \leq j \leq n_x$ and $t \geq 0$, we have 
\begin{align*}
    &\mathbb{P}(w_{t}^{j} + w_{\max} \leq \ell) \geq q_{w}(\ell) > 0,\\
    &\mathbb{P}(w_{\max} - w_{t}^{j} \leq \ell) \geq q_{w}(\ell) > 0.
\end{align*} 
\end{assumption}

\vspace{5pt}

Assumption \ref{ass:w-tight} requires that the boundary  of $\mathcal W$ to be tight in the sense that $w_t$ must visit any small neighborhood of $w_{\max}$ and $-w_{\max}$ with a positive probability. This assumption is necessary for the convergence of SME \cite{akccay2004size,li2024icml}. There are variants of SME proposed in the literature to learn a tight disturbance bound and identify the system uncertainty set together \cite{lu2019robust,fogel1982value}. But this paper only studies the convergence rate of the original  SME in \eqref{eq:SM-eq} for LPN systems and leaves the SME variants for future work.

With Assumption \ref{ass:w-tight}, we can establish the following convergence rate of SME for LPN systems. 

\begin{theorem}[Convergence Rate of SME]\label{th:sme-con-o-l}
Consider Assumptions \ref{ass:analytic}, \ref{ass:bounded-iid-noise}, \ref{ass:liss}, \ref{ass:w-tight}, and i.i.d. random inputs \eqref{equ: iid input}. For any $m \geq 1$ and  $\varrho \in (0, 1)$, when $T > m$, we have
\begin{equation}\label{eq:sme-diam}
\begin{aligned}
&\mathbb{P}\!\big(\textup{diam}(\ThetaSet) > \varrho\big) 
\leq \frac{T}{m}\, \tilde{O}\!\big(n_{\phi}^{2.5}\big)\,a_{2}^{n_{\phi}}\,e^{-a_{3}m} \;\\
&\ \ \ \ \ \ \ \ \ \ +\tilde{O}\!\big(n_{x}^{2.5}n_{\phi}^{2.5}\big)\,a_{4}^{n_{x}n_{\phi}}
\Big(1-q_{w} \!\big(\tfrac{a_{1}\varrho}{4\sqrt{n_{x}}}\big)\Big)^{\lceil T/ m \rceil},
\end{aligned}
\end{equation}
where $\textup{diam}(\ThetaSet)= \sup_{\Theta_1, \Theta_2 \in \ThetaSet}\|\Theta_1-\Theta_2\|_F$,  $a_{1} = {s_{\phi}p_{\phi}}/{4}$, $a_{2} = {64\phi_{\max}^{2}}/({s_{\phi}^{2}p_{\phi}^{2}})$, $a_{3} = {p_{\phi}^{2}}/{8}$, 
$a_{4} = \max({16\phi_{\max}\sqrt{n_{x}}}/({s_{\phi}p_{\phi}}),1)$, and $s_{\phi}, p_{\phi}$ are defined in Theorem \ref{th:BMSB-o-l},  $\phi_{\max}$ is defined in Theorem \ref{thm: LSE open loop}.
\end{theorem}

 For Theorem \ref{th:sme-con-o-l}, the  coefficients hidden in the $\tilde{O}(\cdot)$ notation and the proof are provided in Appendix~\ref{appendix: SME}. 

SME converges if the uncertainty set converges to a singleton $\{\Theta_*\}$.
Theorem \ref{th:sme-con-o-l} analyzes the convergence rate of SME by bounding the diameter of the uncertainty set $\ThetaSet$.   $\textup{diam}(\ThetaSet) $ can be viewed as an upper bound on the estimation error  $\|\hat \Theta - \Theta_*\|_2$ considered in LSE's convergence rate analysis. This is  because $\Theta_* \in \ThetaSet$ and any $\hat \Theta_{\text{SME}}\in \ThetaSet$ can be viewed as a  point estimator of SME, so the estimation error of SME can be upper bounded by the diameter: $\|\hat \Theta_{\text{SME}} - \Theta_*\|_2 \leq \|\hat \Theta_{\text{SME}} - \Theta_*\|_F \leq \textup{diam}(\ThetaSet) $.

Further, Theorem \ref{th:sme-con-o-l} establishes an upper bound on the ``failure'' probability of SME, i.e., the probability that the uncertainty set's diameter exceeds $\varrho$. To ensure the failure probability is less than $1$, one can select $m=O(\log T)$ and choose a sufficiently large $T$ such that $T \geq m=O(\log T)$. If $w_{t}$ is more likely to visit the boundaries of the set $\mathcal{W}$ (meaning a larger $q_w(\ell)$), SME is less likely to estimate an uncertainty set with a diameter greater than $\varrho$. 
To provide more intuition, we consider an example of $q_{w}(\ell) = c_{w} \ell$ for some $c_{w} > 0$. Note that several common distributions, including the uniform distribution and the truncated Gaussian distribution, satisfy this property on $q_{w}(\ell)$ (see Appendix~\ref{appendix: SME} for explicit formulas of $c_w$). With $q_{w}(\ell) = c_{w} \ell$, we can provide a convergence rate of SME in terms of $T$ below.

\begin{corollary}[SME's convergence rate when $q_{w}(\ell) = c_{w} \ell$]\label{cor:sme-rate}Under the 
 Assumptions \ref{ass:analytic}, \ref{ass:bounded-iid-noise}, \ref{ass:liss}, \ref{ass:w-tight}, and i.i.d. random inputs \eqref{equ: iid input}, consider $w_t$'s distribution satisfying $q_{w}(\ell) = c_{w} \ell$ for all $0<\ell <\ell_{\max}$ for some $c_w>0, \ell_{\max}>0$. For  any $\delta > 0$, let
\begin{equation*}
m \;\ge\; O\!\left( \frac{\log(T/\delta) + n_{\phi} \log\!\big(\tfrac{8 \phi_{\max}}{s_{\phi}p_{\phi}}\big)}{p_{\phi}^{2}} \right).
\end{equation*}
then when $T>m$, with probability at least $1-2\delta$,
\begin{equation*}
\mathrm{diam}(\ThetaSet) \le
\tilde O\left( \frac{%
    n_{\phi}^2 n_x^{1.5}\log(1/\delta)}{T } \right).
\end{equation*}
\end{corollary}
\vspace{3pt}

The proof is provided in Appendix~\ref{appendix: SME}. 

Corollary \ref{cor:sme-rate} shows that, for the case of $q_{w}(\ell) = c_{w} \ell$ and under additional Assumption \ref{ass:w-tight}, the convergence rate of SME for LPN systems is $\tilde O(1/T)$, which is better than the $O(1/\sqrt T)$ convergence rate of LSE in terms of the sampling horizon $T$. This is consistent with the convergence rate comparison for linear system identification in \cite{li2024icml}. 

Finally, similar to LSE, we can generalize the convergence rates of SME to closed-loop systems with randomly perturbed policies \eqref{equ: randomly perturbed control}.

\begin{corollary}[SME with Randomly Perturbed Policies]\label{cor:sme-con-c-l}
 Under Assumptions~\ref{ass:analytic},~\ref{ass:bounded-iid-noise}, \ref{ass:w-tight}, consider  a randomly perturbed control policy \eqref{equ: randomly perturbed control} that satisfies Assumption \ref{ass: stabilizing controller}.  
 The convergence rate of SME in Theorem~\ref{th:sme-con-o-l} still holds under the parameters $\tilde s_{\phi}, \tilde p_{\phi}, \tilde \phi_{\max}$ introduced in Corollary \ref{cor: bmsb closed loop}. 
 
 Further, if $w_t$'s distribution satisfies $q_{w}(\ell) = c_{w} \ell$ for all $0<\ell <\ell_{\max}$ for some $c_w>0$, the convergence rate of SME is $\tilde O(1/T)$.
\end{corollary}

\vspace{4pt}
The proof of Corollary \ref{cor:sme-con-c-l} is provided in Appendix~\ref{appendix: SME}.

%% file: CounterExample.tex
\section{More Discussions on Assumptions}\label{sec:counter example}

This section provides more discussions on the assumptions adopted in this paper. In particular, we will focus on the real-analytic condition in Assumption \ref{ass:analytic} and the semi-continuity condition in Assumption \ref{ass:bounded-iid-noise} because they are not standard assumptions in the recent literature of finite-sample analysis of system identification. We will provide counter-examples to show the importance of these two conditions.

\subsection{More Discussions on Real-Analytic Condition}
This subsection provides more discussions on Assumption \ref{ass:analytic} and constructs a counter-example that satisfies all the assumptions in Theorem \ref{thm: LSE open loop} except for Assumption \ref{ass:analytic}. We will show that i.i.d. random exploration is  insufficient to identify this counter-example system, at least under bounded  noises. 

First, we review  the existing  results and discuss the remaining gap. Recall that 
\cite{mania2022active} has provided a counter-example by a piecewise affine system to show that non-active exploration is insufficient to learn/identify LPN systems, while bilinear systems can be identified by non-active exploration \cite{sattar2022finite}. Notice that  the piecewise-affine functions are \textit{not differentiable}, yet the bilinear systems are \textit{infinitely differentiable} and even \textit{real-analytic}. This paper bridges the gap by showing that real-analytic LPN systems can be identified by non-active exploration. However, there is still a gap between non-differentiable functions and real-analytic functions, which naturally leads to the following question: is a higher-order of differentiability sufficient to guarantee successful identification by non-active exploration?

To address this question, we construct a counter-example system that is infinitely differentiable but not real-analytic. This is based on a standard example of an infinitely differentiable function that is not real analytic, which is reviewed below:
\begin{equation}\label{equ: non real analytic function}
    f(x)=
\begin{cases}
e^{-1/x^{2}}, & x>0,\\
0, & x\le0,
\end{cases}
\end{equation}
where $f(x)$ is infinitely differentiable everywhere but not real-analytic at \(x=0\). Based on function \eqref{equ: non real analytic function}, we construct the following system:
\begin{equation}\label{equ: counter example 1}
    x_{t+1}= a_* f(x_t - 4)+ b_*u_t + w_t
\end{equation}
where $a_*=b_*=1$ are unknown parameters, $w_t $ follows a uniform distribution on $[-1,1]$ and $u_t=\eta_t$ follows a uniform distribution on $[-2,2]$. It is straightforward to verify that counter-example above satisfies Assumptions \ref{ass:bounded-iid-noise} and \ref{ass:liss}, but not Assumption \ref{ass:analytic}.

Next, we show that $a_*$ cannot be identified under i.i.d. $u_t=\eta_t$ following uniform distribution on $[-2,2]$. To learn $a_*$, $x_t$ should be larger than 4, otherwise $f(x_t-4)$ is a constant 0. However, we can show $x_t\leq 3$ for all $t\geq 0$ by induction below. Starting from $x_0=0$, if $x_t \leq 3$ at certain $t\geq 0$, we have $x_{t+1} =a_* f(x_t-4)+b_*u_t+w_t= u_t +w_t\leq 3$. Therefore, $x_t \leq 3$ for all $t$, so $a_*$ can not be identified. Figure \ref{fig:non_analytic_lse_sme_combo} visualizes our theoretical insight and shows that the estimation error of $b_*$ can converge to 0, but the estimation error of $a_*$ remains a constant (the initial guess of $a_*$) for both LSE and SME. 

\begin{remark}
Our counter example \eqref{equ: counter example 1}  is constructed for the bounded-noise case because this paper focuses on bounded noises for the convergence rate analysis, which is consistent with a lot of existing work on general nonlinear system learning \cite{mania2022active,li2023online,lin2024online} and is practical because many real world systems only have bounded noises. 

However, since unbounded Gaussian noises are commonly considered in linear and bilinear system identification, it is also worth discussing the counter-example under unbounded  noises. In this case,  the counter-example in \eqref{equ: counter example 1} no longer applies directly because $w_t$ and $u_t$ have a positive probability of driving the state to be $x_t> 4$. 

Nevertheless, by constructing a multi-dimensional system, we can construct a counterexample for which informative states are visited with exponentially small probability in the state dimension, resulting in an exponentially large sample complexity for non-active exploration. In particular, consider $x_{t+1}=a_* f(\|x_t \|_2- n_x) \mathbf{1}_{n_x}+b_* u_t + w_t$, where $a_*=b_*=1$ are unknown parameters, $\mathbf{1}_{n_x}$ is a known all-one vector, $x_t \in \R^{n_x}$, $u_t$ and $w_t$ are independent and follow standard Gaussian distributions $N(0,I_{n_x})$. The feature function $f(\|x_t \|_2- n_x)$ is infinitely differentiable because: away from the origin this follows from the smoothness of the Euclidean norm and the chain rule, while in a neighborhood of the origin the feature map is identically zero. Moreover, constant-zero region implies that, to identify $a_*$, the state should satisfy $\|x_t\|_2>n_x$.  This can only be achieved when $\|u_t +w_t\|_2 \geq n_x - \sqrt{n_x}$ because  the nonlinear component $\| a_* f(\|x_t \|_2- n_x)\mathbf{1}_{n_x}\|_2 \leq \sqrt{n_x}$. By the concentration inequality of Gaussian distributions, $\Pb(\|u_t +w_t\|_2 \geq n_x - \sqrt{n_x}) \leq e^{-c n_x}$ for some universal constant $c>0.$ Therefore, under i.i.d. random inputs, the state reaches the informative region only with exponentially small probability, in other words, the sample complexity is exponentially large in terms of $n_x$, indicating the insufficiency of non-active exploration in this case. 

The corresponding sample complexity for real-analytic functions under Gaussian noises is left for future studies. 

\end{remark}

\begin{figure}[ht]
    \centering

    \begin{subfigure}[t]{0.46\linewidth}
        \centering
        \includegraphics[width=\linewidth]{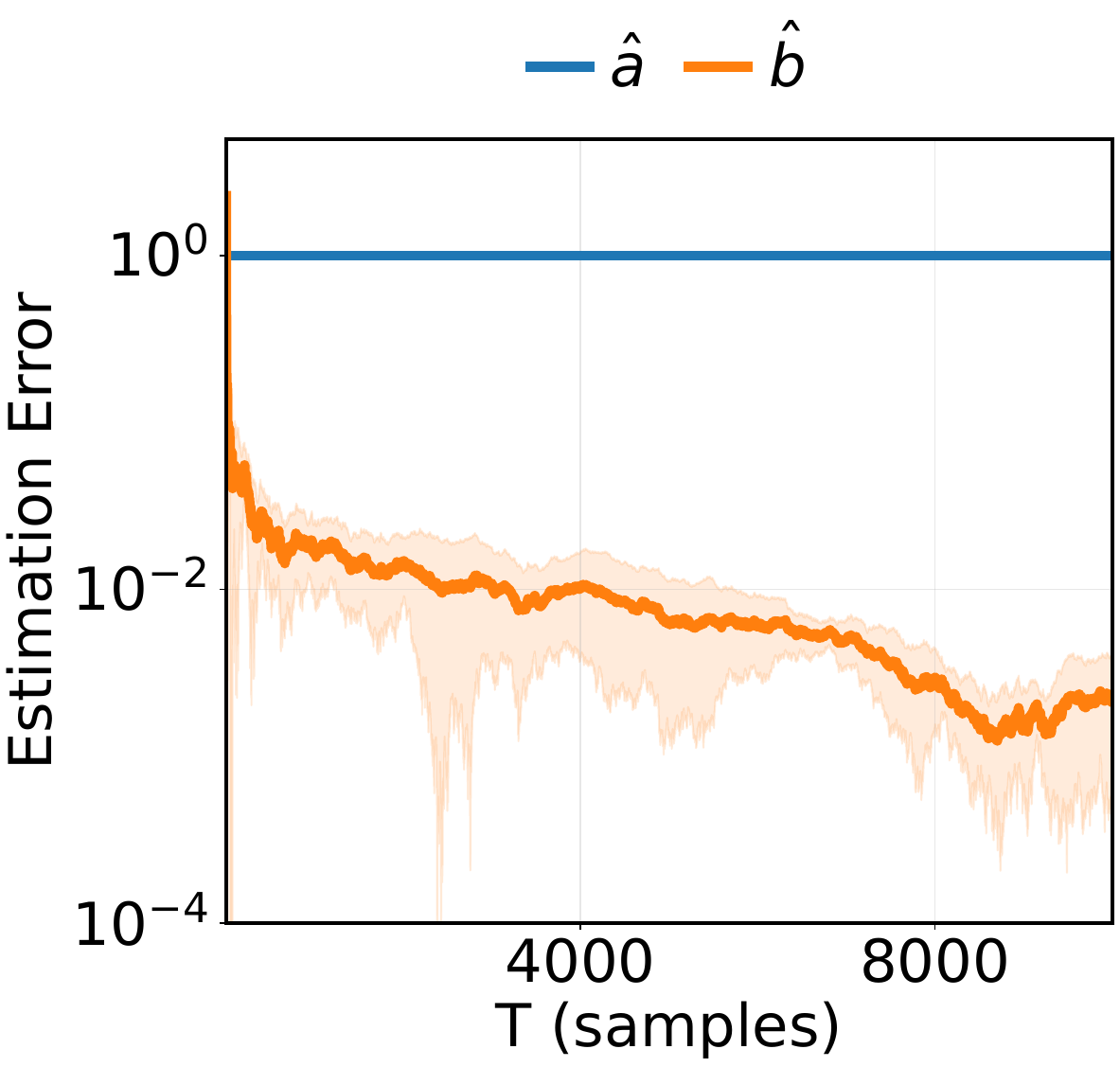}
        \caption{LSE}
        \label{fig:non_analytic_lse}
    \end{subfigure}
    \hfill
    \begin{subfigure}[t]{0.50\linewidth}
        \centering
        \includegraphics[width=\linewidth]{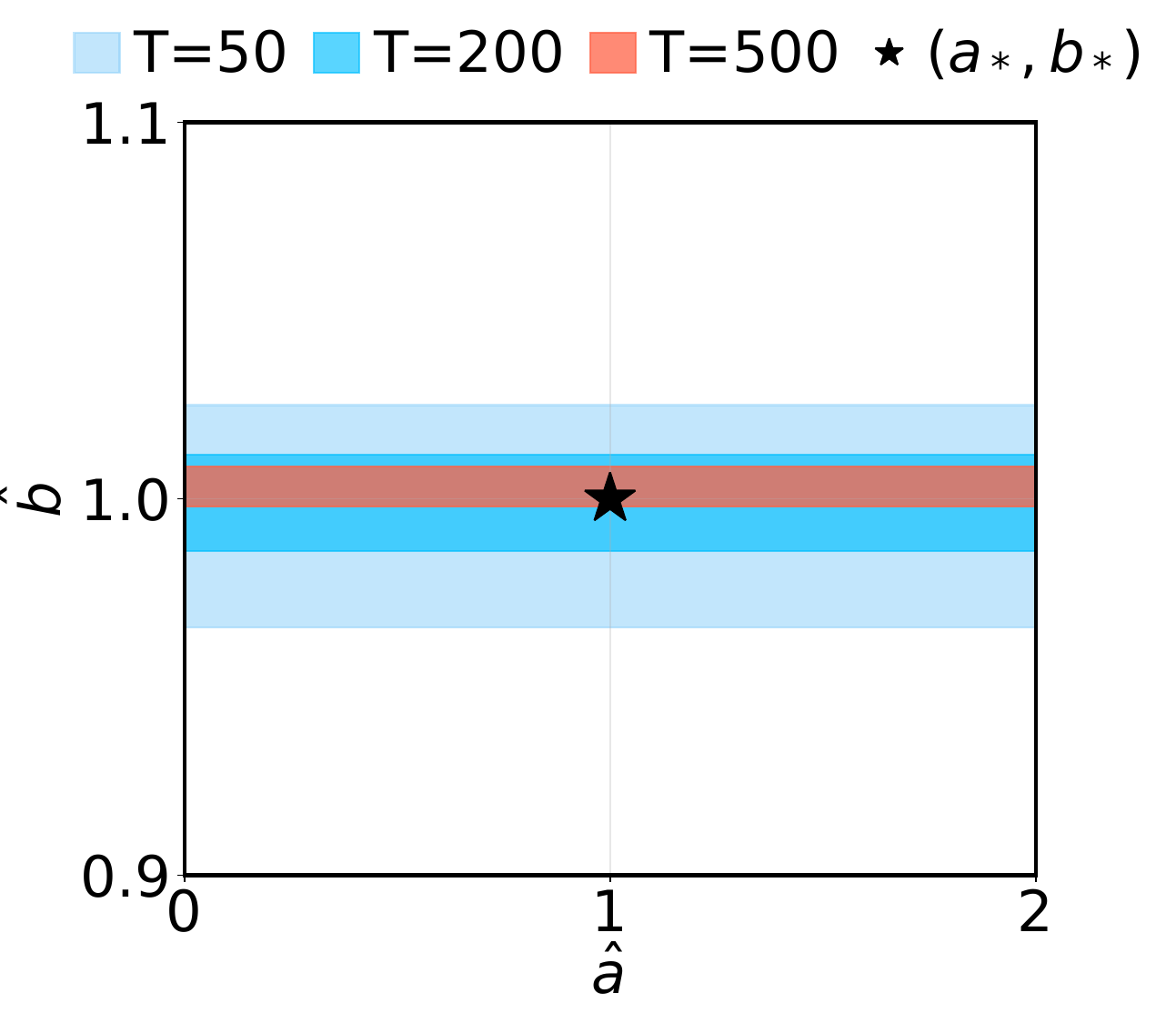}
        \caption{SME}
        \label{fig:non_analytic_sme}
    \end{subfigure}

    \caption{Illustration of counter-example \eqref{equ: counter example 1} for a nonlinear system with non-real-analytic yet infinitely differentiable feature functions.
    (a) Estimation errors of LSE.
    (b) Uncertainty regions constructed by SME.
    Consider an initial estimate of $\hat a_0=\hat b_0=0$ and an initial uncertainty size of 1 for visualization purposes.
    In both cases, the estimation error of $a_*$ remains constant and does not improve with $T$.
    }
    \label{fig:non_analytic_lse_sme_combo}
\end{figure}

\subsection{More Discussions on Semi-Continuous Condition}

This subsection provides more discussions on the semi-continuity condition in Assumption \ref{ass:bounded-iid-noise}. In particular, we construct counter examples that satisfy all the assumptions in Theorem \ref{thm: LSE open loop} except for the semi-continuity condition and shows that system parameters cannot be identified by non-active exploration in these counter examples.

 Consider the following system
\begin{align*}
    \begin{pmatrix}
        \alpha_{t+1}\\
        \beta_{t+1}
\end{pmatrix}=\begin{pmatrix}
        a_* & 0 & 0\\
         b_* &  c_* & d_*
    \end{pmatrix} \begin{pmatrix}
        \sin(\alpha_t)\\
        \beta_t \\
        \sin(u_t)
    \end{pmatrix}+  \begin{pmatrix}
        \epsilon_{t}\\
        \nu_{t}
\end{pmatrix}
\end{align*} 
with state \(x_t=(\alpha_t,\beta_t)^\top\), feature functions $\phi(x_t, u_t)= (\sin(\alpha_t), \beta_t, \sin(u_t))^\top$, disturbance $w_t= (\epsilon_t, \nu_t)^\top$, $a_*=b_*=c_*=1/2, d_*=1$, and  $x_0=0$. It is straightforward to verify Assumptions \ref{ass:analytic} and \ref{ass:liss} for this system.

First, let's focus on the semi-continuity condition  on the disturbances. Consider the following disturbances and inputs. Let $\epsilon_t$ follow a uniform distribution on $\{-\pi, \pi\}$. Let $\nu_t$ and $u_t$ follow uniform distributions on $[-1,1]$. Therefore, all the conditions in Assumption \ref{ass:bounded-iid-noise} are satisfied except for the semi-continuity condition on $w_t$. Notice that $a_*$ cannot be identified in this case because $\alpha_t \in \{-\pi, \pi\}$ for all $t\geq 1$ so $a_*\sin(\alpha_t)=0$  for all $t$, no matter what the value of $a_*$ is. Similarly, $b_*$ cannot be identified because $b_*\sin(\alpha_t)=0$  for all $t$, no matter what the value of $b_*$ is. Therefore, this counter example has shown that the lack of semi-continuity on disturbances may prevent the identification of system parameters. 

Next, we focus on the semi-continuity condition on control inputs. Consider $w_t$ to follow a uniform distribution on a unit ball, and $u_t$ to follow a uniform distribution on $\{-\pi, \pi\}$. It is straightforward to see that all conditions in Assumption \ref{ass:bounded-iid-noise} are satisfied except for the semi-continuity condition on $u_t$. In this case, $d_*$ cannot be identified because $d_* \sin(u_t)=0$ for all $t\geq 0$ no matter what the value of $d_*$ is. Therefore, this counter example has shown that the lack of semi-continuity on control inputs may prevent the identification of system parameters. 

Figure~\ref{fig:assmp_lse_sme_combo} visualizes the two counter-examples above, where LSE's estimation errors for certain parameters do not decrease to zero, and  SME's uncertainty sets fail to shrink along the unidentifiable parameter directions.

\begin{figure}[ht]
    \centering

    \begin{subfigure}[t]{0.48\linewidth}
        \centering
        \includegraphics[width=\linewidth]{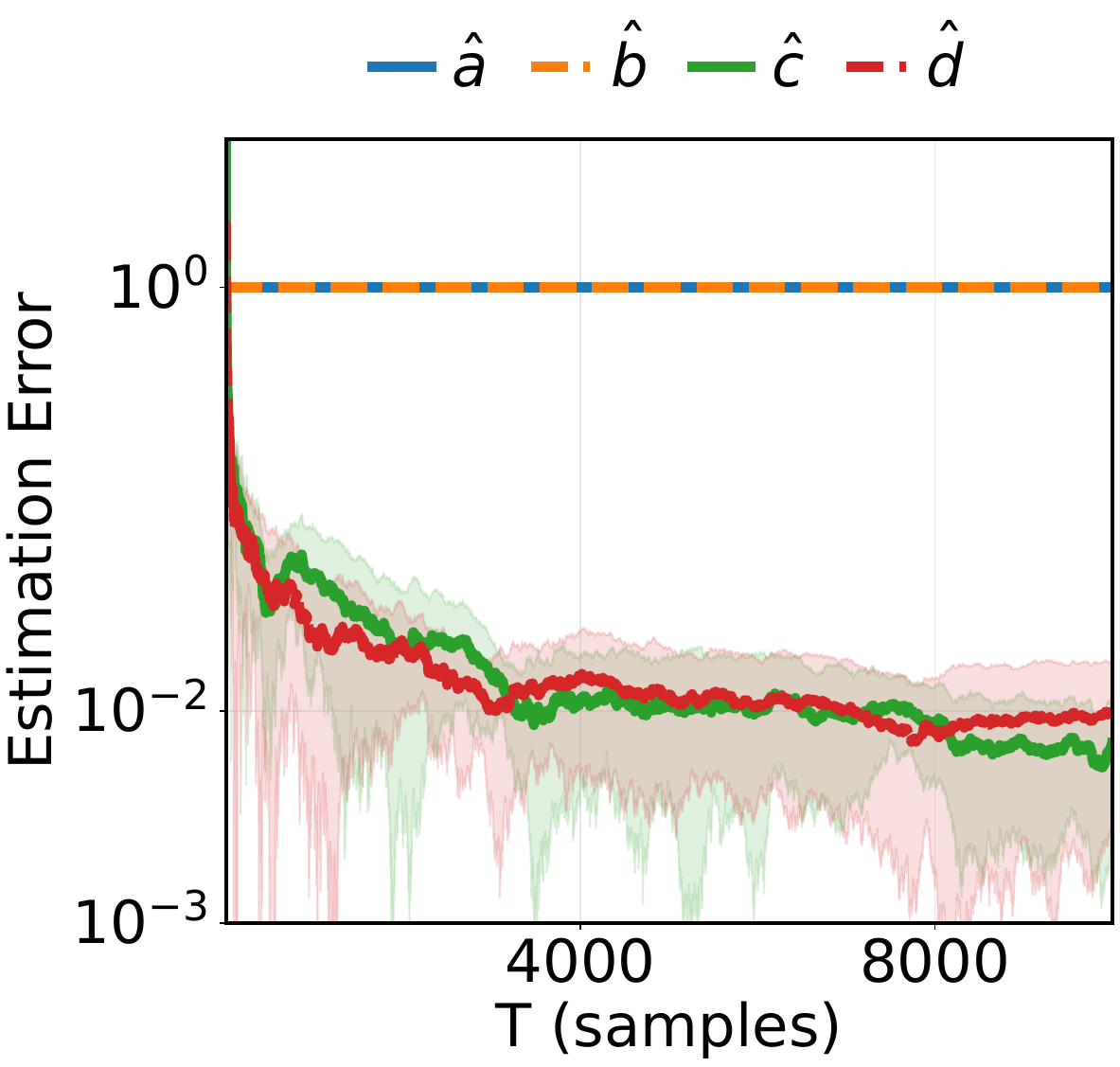}
        \caption{LSE}
        \label{fig:assmp_disturbance_lse}
    \end{subfigure}
    \hfill
    \begin{subfigure}[t]{0.50\linewidth}
        \centering
        \includegraphics[width=\linewidth]{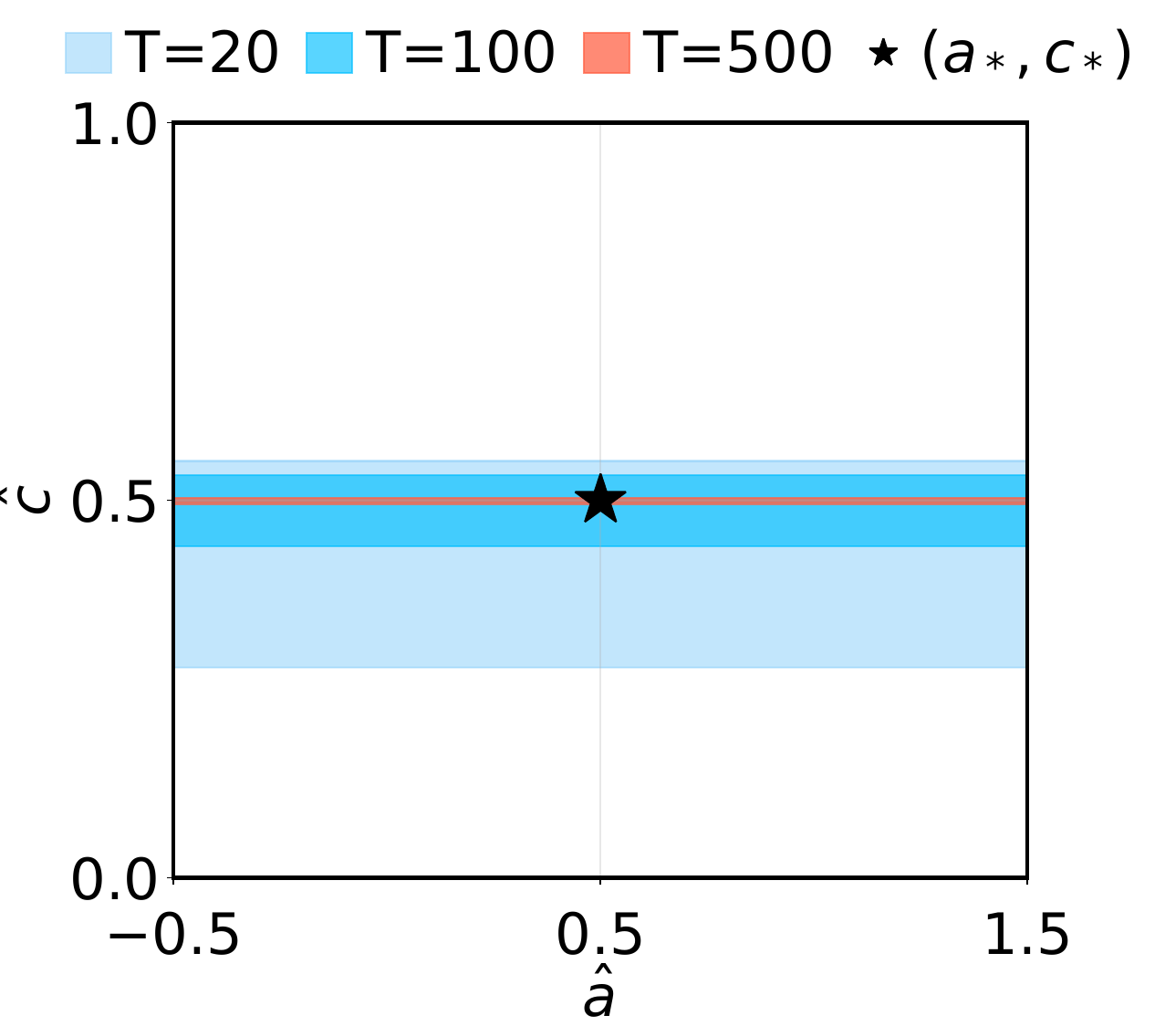}
        \caption{SME}
        \label{fig:assmp_disturbance_sme}
    \end{subfigure}

    \vspace{0.3cm}

    \begin{subfigure}[t]{0.48\linewidth}
        \centering
        \includegraphics[width=\linewidth]{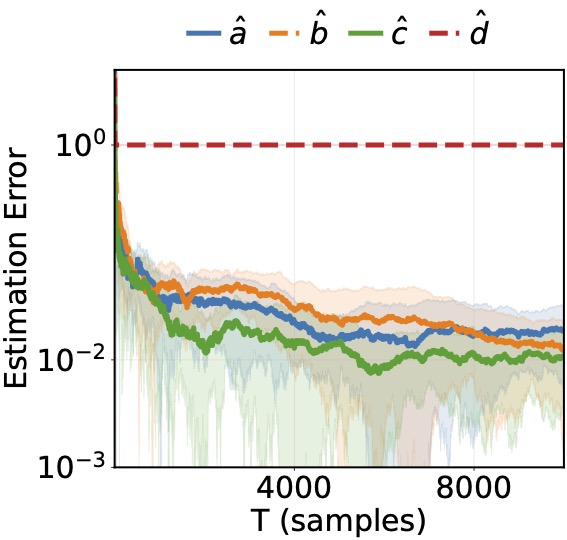}
        \caption{LSE}
        \label{fig:assmp_input_lse}
    \end{subfigure}
    \hfill
    \begin{subfigure}[t]{0.50\linewidth}
        \centering
        \includegraphics[width=\linewidth]{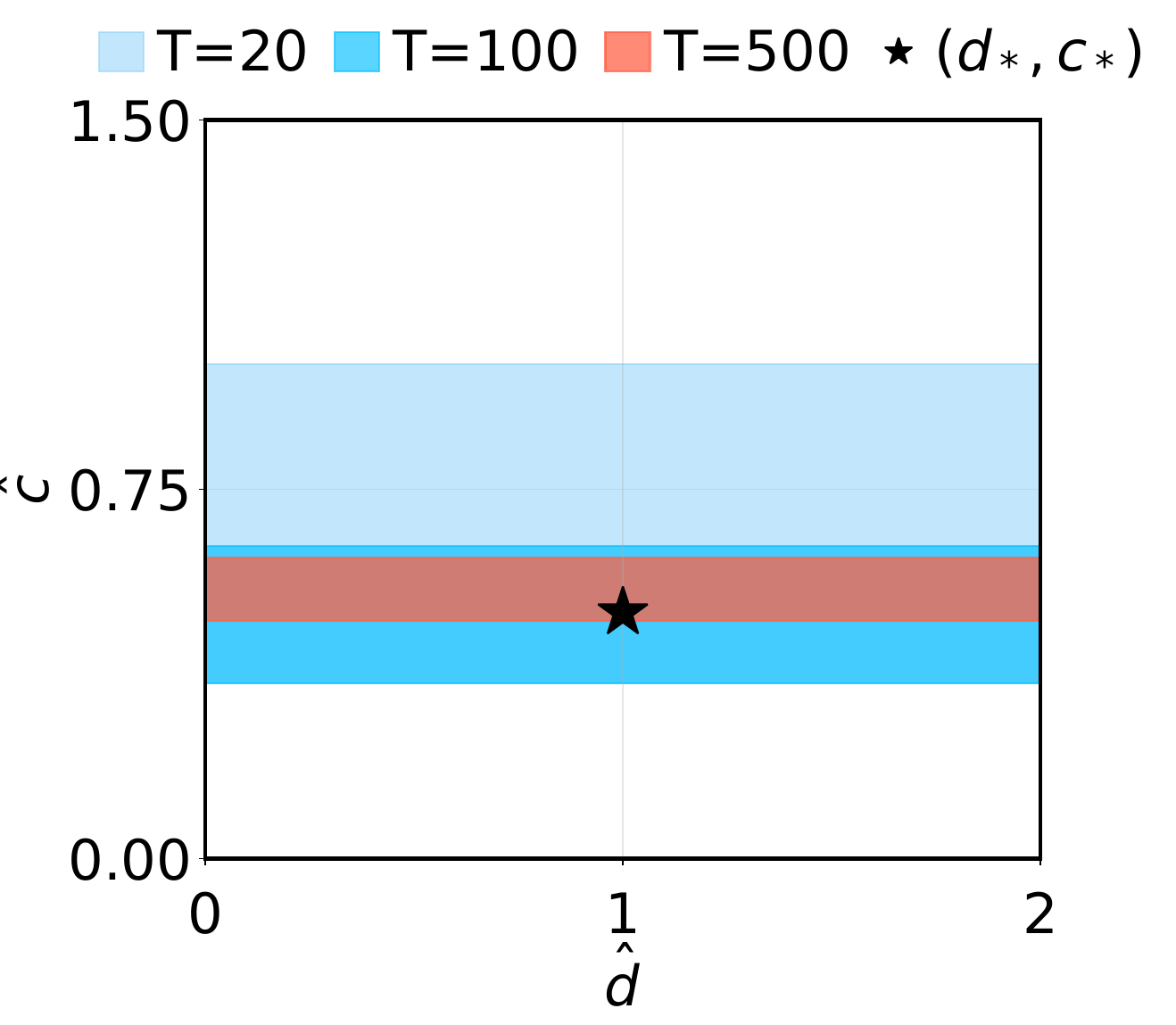}
        \caption{SME}
        \label{fig:assmp_input_sme}
    \end{subfigure}

    \caption{Failure of parameter identification when the semi-continuity condition is violated.
    Figures (a) and (b) correspond to the counterexample with non-semi-continuous $w_t$, and
    Figures (c) and (d) correspond to the counterexample with non-semi-continuous $\eta_t$.
    }
    \label{fig:assmp_lse_sme_combo}
\end{figure}

%% file: Numerical_Examples.tex
\section{Numerical Results}\label{sec:num-exp}

This section provides numerical experiments to demonstrate our theoretical results. The code  is publicly available on GitHub.\footnote{\url{https://github.com/ZiyaoG/real-analytic-nonlinear-sys-id-main.git}
}

\vspace{-10pt}
\subsection{Experiment  Settings}
\textbf{Pendulum} 
Following Example~\ref{ex:pendulum} from Section~\ref{sec:prob}, we let $
\theta_{1} = \frac{1}{l},  \theta_{2} = \frac{1}{Ml^{2}},
$
where the true values are \( M = 0.1\ \text{(kg)} \) and \( l = 0.5\ \text{(m)} \). The control input is defined as $u_t = -k_d\, \omega_t + \eta_t,$ where \( k_d \) is a feedback gain and \( \eta_t \) is i.i.d. noise. In our experiments, we use \( k_d = 2 \) for Figures~\ref{fig:lse_pend_uni} and~\ref{fig:lse_pend_trunc}, and \( k_d = 0.1 \) for Figures~\ref{fig:sme_pend_uni},~\ref{fig:sme_pend_trunc},~\ref{fig:pend_sme}, \ref{fig:error_active} and Table~\ref{tab:active_nonactive_comparison}. The discretization time step used in all simulations is \( \Delta_T = 0.01\ \text{(s)} \).

\textbf{Drone} 
Consider  Example~\ref{ex:quadrotor} in Section~\ref{sec:prob}, where the state dimension is 13 and the control input dimension is 4. The control input is given by \( u_t = \pi(x_t) + \eta_t \), where \( \pi(x_t) \) regulates altitude and the three Euler angles, and is adapted from~\cite{alaimo2013mathematical} with the following controller gains $
    kp_z = 0.75,\ kd_z = 1.25,\
    kp_{att}=0.03, \ kd_{att}=0.00875,
$

where the $att$ subscript indicates that the same attitude gains are applied to roll, pitch and yaw. 
The unknown parameter vector \( \theta_* \in \mathbb{R}^7 \) encodes physical quantities such as the mass and components of the inertia matrix. We adopt a diagonal inertial matrix for simplicity. The ground truth values for these parameters are $
    M = 0.468\ \text{(kg)},
    I_{xx} = I_{yy} = 4.856\times 10^{-3}, I_{zz} = 8.801\times 10^{-3}\ \text{(kg$\cdot$m$^2$)}.
$
These yield the following seven unknown parameters:
$
    \theta_1 = \frac{1}{M},\
    \theta_2 = \frac{1}{I_{xx}},\
    \theta_3 = \frac{I_{yy} - I_{zz}}{I_{xx}},\
    \theta_4 = \frac{1}{I_{yy}},\
    \theta_5 = \frac{I_{zz} - I_{xx}}{I_{yy}},\
    \theta_6 = \frac{1}{I_{zz}},\
    \theta_7 = \frac{I_{xx} - I_{zz}}{I_{zz}}.
$
We use a discretization time step of \( \Delta_T = 0.01 \) seconds. 

\subsection{Convergence Rates of LSE}
Figures~\ref{fig:lse_pend_uni} and~\ref{fig:lse_pend_trunc} present a comparison between  LSE's theoretical bound from Theorem~\ref{thm: LSE open loop} with its empirical estimation error of the unknown parameters $\theta_{*}$ versus trajectory length $T$ for the pendulum example, with uniform and truncated-Gaussian noises and disturbances. Similarly, Figures~\ref{fig:lse_quad_uni} and~\ref{fig:lse_quad_trunc} show this comparison for the drone example. Here, uniform noises and disturbances are i.i.d. generated from $\texttt{uniform}([-1, 1])$, and truncated-Gaussian noises and disturbances are i.i.d. generated from $\texttt{trunc-G}(0, 0.1, [-1, 1])$. 

In each figure, both the theoretical bound and the empirical error are normalized by the $l_2$ norm of the nominal parameter $\theta_{*}$. The empirical errors are averaged across 20 trials and standard deviation is reflected by the shaded area.  To numerically visualize the theoretical bound, we estimate the BMSB parameters by the numerical procedures  described in
Appendix~\ref{append: numerical}. The log-log plots for both
scenarios show that the empirical error decays at a rate of
$O(1/\sqrt{T})$, which is consistent with the theoretical rate in
Theorem~\ref{thm: LSE open loop}. The theoretical bound is larger than the empirical errors by a constant factor due to the conservativeness of  theoretical analysis. 

\begin{figure}[ht]
    \centering
    \subcaptionbox{Uniform\hspace*{-0.75cm} \label{fig:lse_pend_uni}}[0.24\textwidth]{      
        \includegraphics[width=\linewidth]{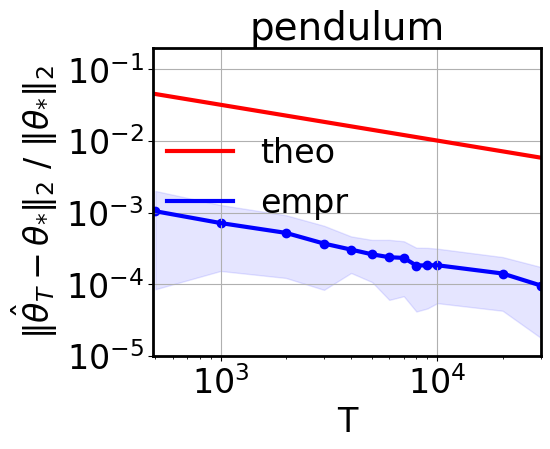}
    }
    \subcaptionbox{Truncated-Gaussian\hspace*{-0.4cm}\label{fig:lse_pend_trunc}}[0.24\textwidth]{   
        \includegraphics[width=\linewidth]{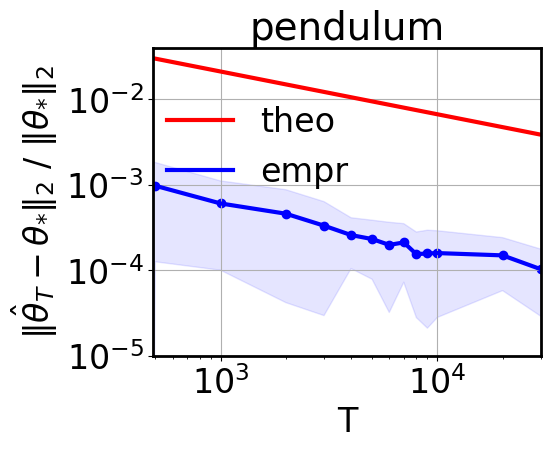}
    }
    \subcaptionbox{Uniform\hspace*{-0.9cm}\label{fig:lse_quad_uni}}[0.24\textwidth]{     
        \includegraphics[width=\linewidth]{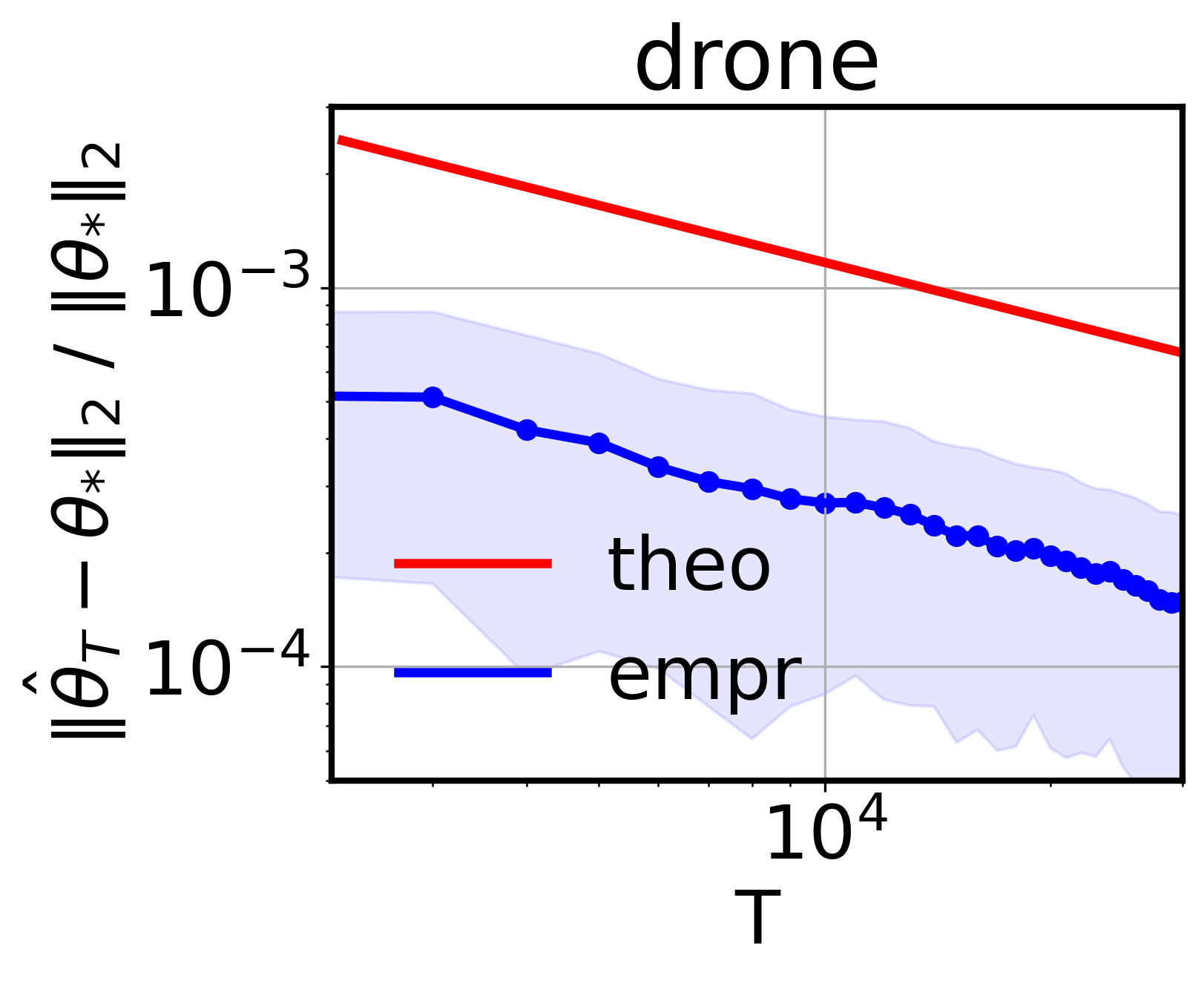}
    }
    \subcaptionbox{Truncated-Gaussian\hspace*{-0.75cm}\label{fig:lse_quad_trunc}}[0.24\textwidth]{
        \includegraphics[width=\linewidth]{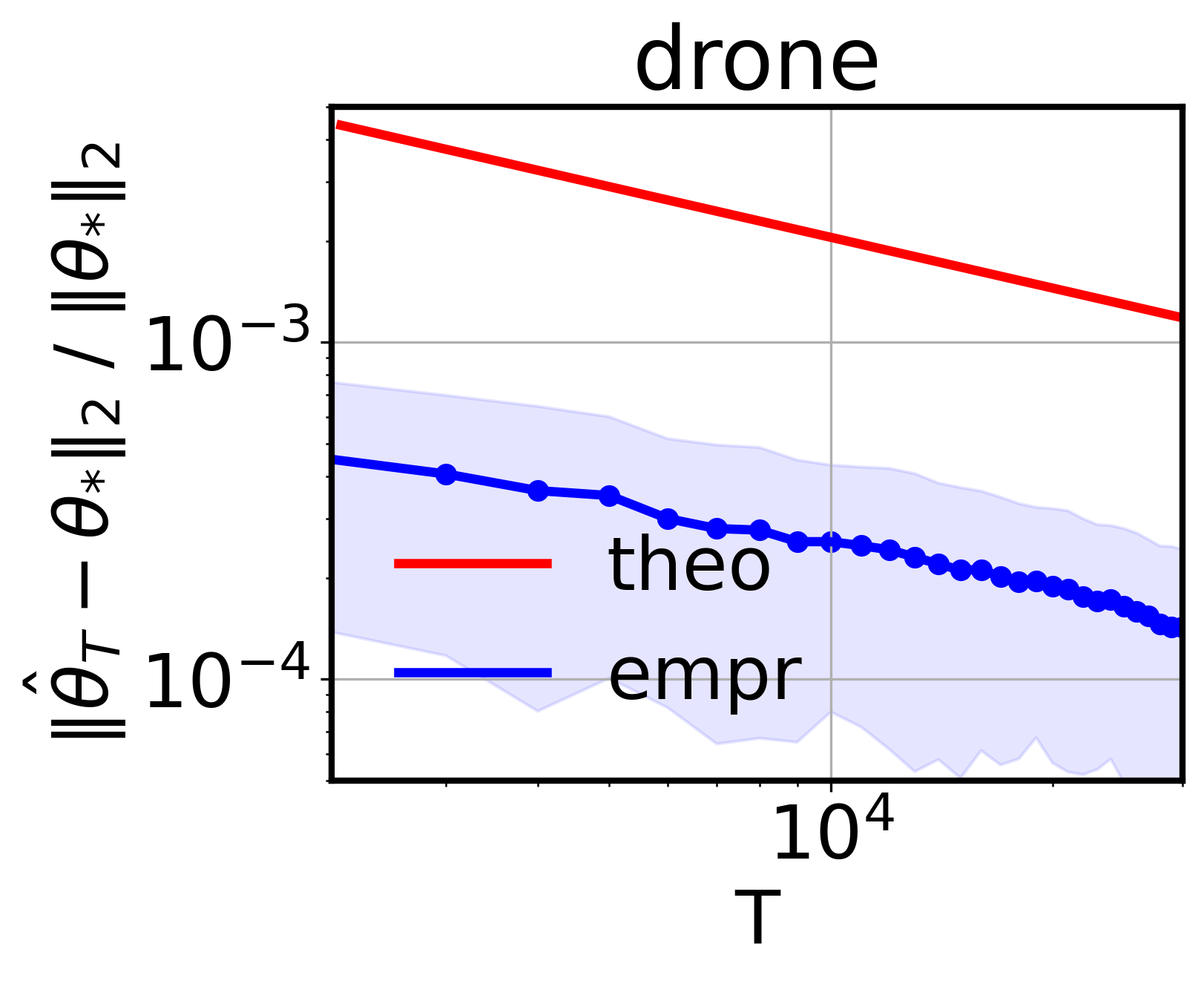}
    }
    \caption{Convergence rate of the LSE for pendulum and drone scenarios.
     "theo" denotes the theoretical convergence rate, and "empr" represents the empirical rate. The shaded areas illustrate empirical standard deviation.}    
\end{figure}

\subsection{Convergence Rates of SME}
Figures~\ref{fig:sme_pend_uni} and~\ref{fig:sme_pend_trunc} show the empirical convergence rates of SME for the pendulum example, for uniform and truncated-Gaussian noises and disturbances, in comparison to the theoretical rates from Corollary~\ref{cor:sme-rate}. Here, uniform noises and disturbances are i.i.d. generated from $\texttt{uniform}([-1, 1])$, and truncated-Gaussian noises and disturbances are i.i.d. generated from $\texttt{trunc-G}(0, 0.5, [-1, 1])$. In each figure, both the theoretical bound and the empirical error are normalized by the $l_2$ norm of the nominal parameter $\theta_{*}$. The empirical errors are averaged across 10 trials and standard deviation is reflected by the shaded area.  

The log-log plots indicate that the empirical rates achieve $O(\frac{1}{T})$, which is consistent with the results from Corollary~\ref{cor:sme-rate} and with the related results for linear systems in~\cite{li2024icml}. A similar result can be observed for the drone cases in Figures~\ref{fig:sme_quad_uni} and~\ref{fig:sme_quad_trunc}. 

Further, Figure~\ref{fig:pend_sme} shows the uncertainty sets estimated by SME for the two unknown parameters, labeled \( \theta_{1} \) and \( \theta_{2} \), in the pendulum example, along with the diameters of these sets as trajectory length grows. We observe that these sets contract as trajectory length increases, with the true values of the unknown parameters lying within the estimated uncertainty sets. Similarly, 
Figure~\ref{fig:sm_drone_2} displays the uncertainty set estimated by SME for  the drone example for various trajectory lengths, with \(\eta_{t}\) and \(w_{t}\) being i.i.d. samples from truncated-Gaussian distributions. The uncertainty sets are observed to shrink as the trajectory length increases and the ground truth  is contained within all the uncertainty sets.

\begin{figure}[htp]
    \centering
    \subcaptionbox{Uniform\hspace*{-0.75cm} \label{fig:sme_pend_uni}}[0.24\textwidth]{     
        \includegraphics[width=\linewidth]{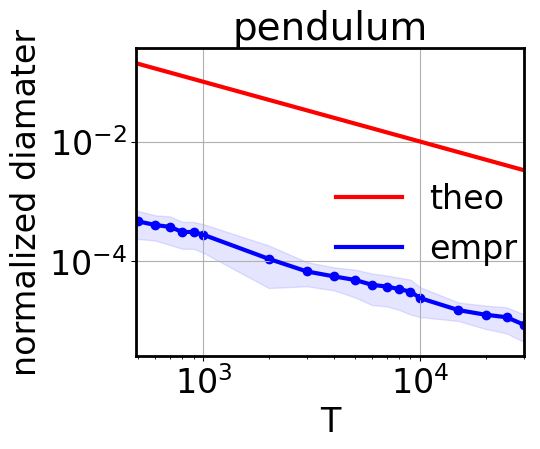}
    }
    \subcaptionbox{Truncated-Gaussian\hspace*{-0.4cm}\label{fig:sme_pend_trunc}}[0.24\textwidth]{  
        \includegraphics[width=\linewidth]{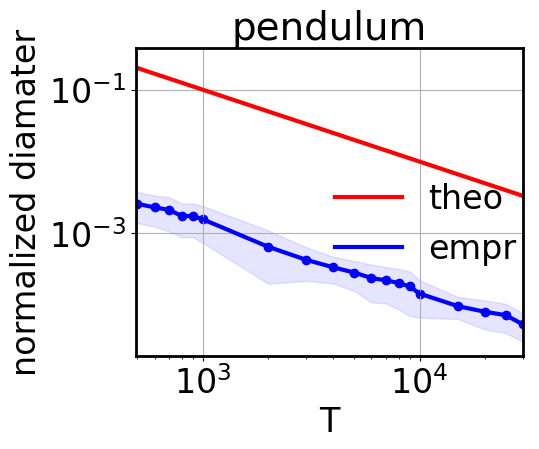}
    }
    \subcaptionbox{Uniform\hspace*{-0.9cm}\label{fig:sme_quad_uni}}[0.24\textwidth]{     
        \includegraphics[width=\linewidth]{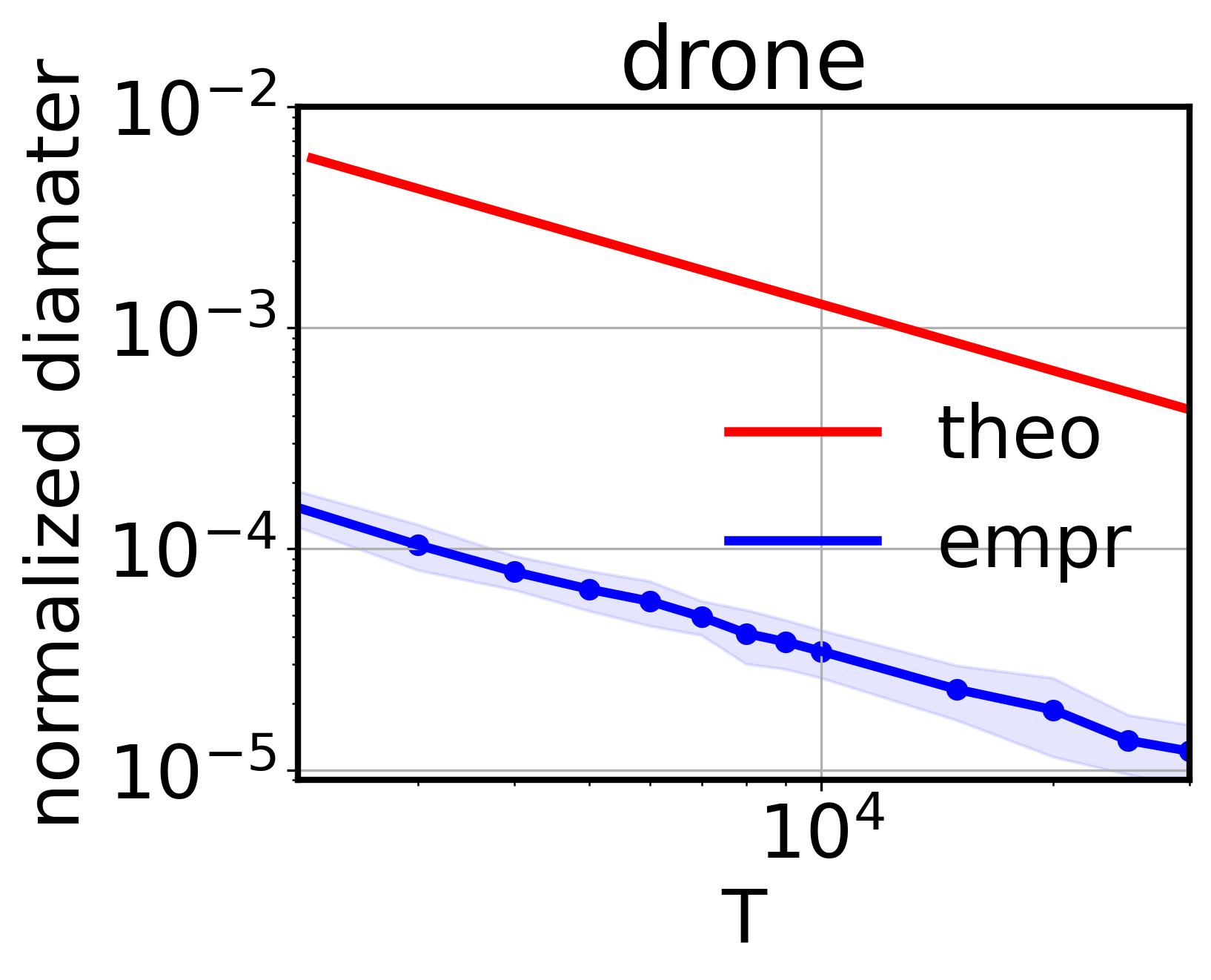}
    }
    \subcaptionbox{Truncated-Gaussian\hspace*{-0.75cm}\label{fig:sme_quad_trunc}}[0.24\textwidth]{  
        \includegraphics[width=\linewidth]{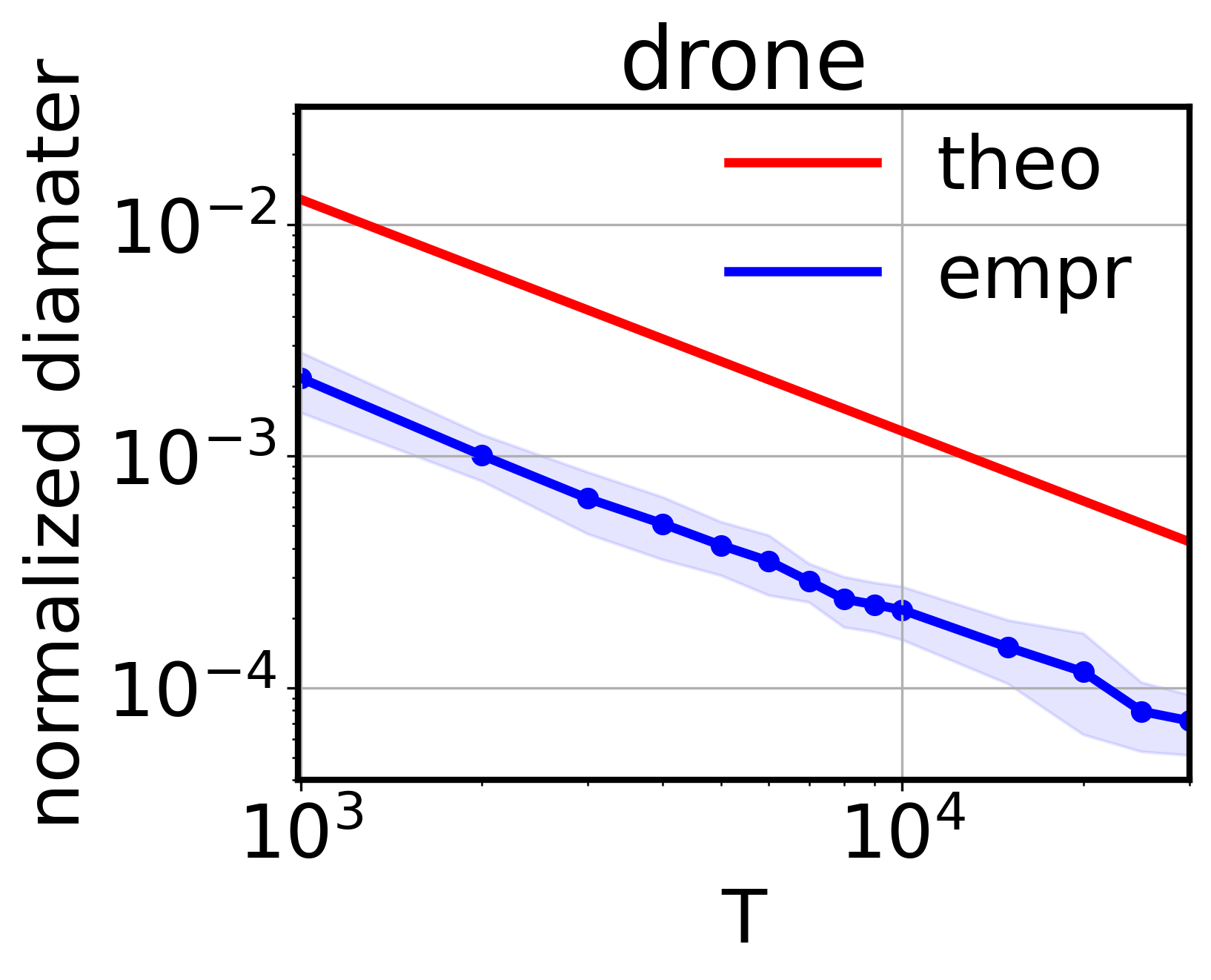}
    }
    \caption{Convergence rate of the SME for pendulum and drone scenarios.  "theo" denotes the theoretical convergence rate, and "empr" represents the empirical rate. }
\end{figure}

\begin{figure}[!htp]
    \subcaptionbox{Uncertainty set diameter\hspace{1.5cm}(b) Uncertainty set\hspace*{0.0cm}}[0.5\textwidth]{      
        \includegraphics[width=0.98\linewidth]{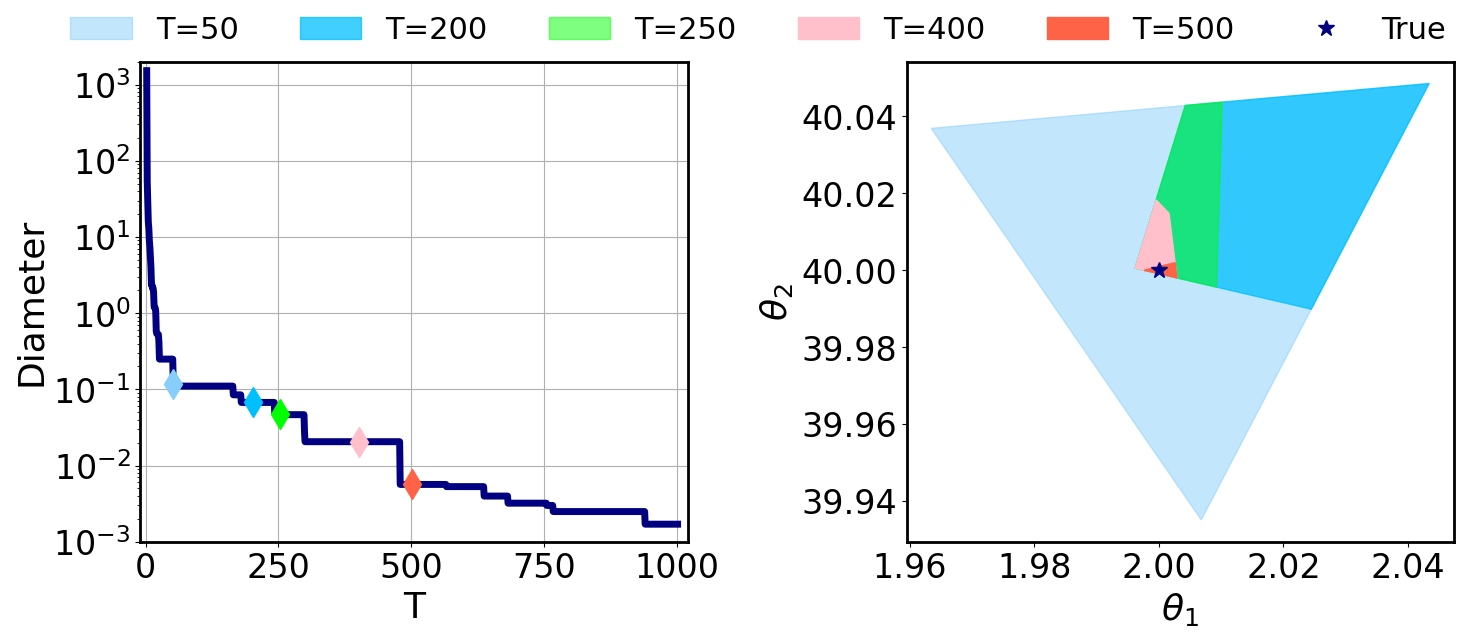}
    }
    \caption{Performance of SME for pendulum with controller $u_{t} = -k_d\omega_{t}+\eta_{t}$, where $\eta_{t}$ i.i.d. generated from $\texttt{trunc-G}(0, 2, [-2, 2])$. System disturbances $w_{t}$ are i.i.d. generated from $\texttt{trunc-G}(0, 1, [-1, 1])$. (a) Diameter of the uncertainty set estimated by SME. (b) Uncertainty set depicted for $T=50, 200, 250, 400, 500$.}
    \label{fig:pend_sme}
\end{figure}

\begin{figure}[ht]
    \centering    \includegraphics[width=0.98\linewidth]{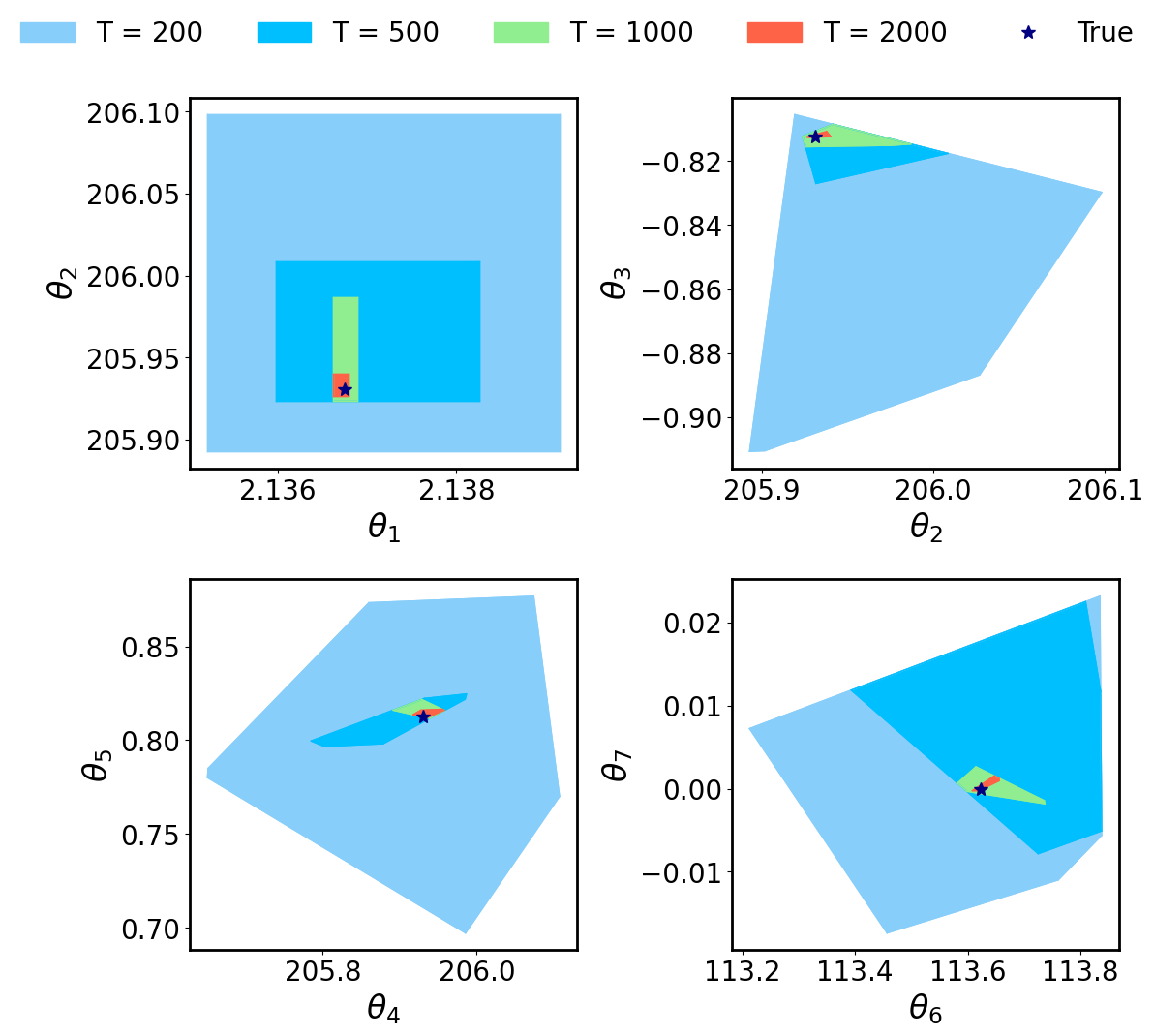}
    \caption{2D projections of the uncertainty set estimated by SME for the unknown parameters of the drone example. The noises and disturbances are i.i.d generated from $\texttt{trunc-G}(0,0.5,[-1,1])$. 
    }
    \label{fig:sm_drone_2}
\end{figure}

\subsection{Numerical Comparison with Other System ID Methods}

Finally, we compare LSE and SME with other system identification methods in the literature. First, we compare with an active exploration method in \cite{mania2022active}  to illustrate a central message of this paper: for LPN systems with real-analytic feature functions, non-active exploration already generates sufficient excitation for efficient identification.  In addition, we compare with a computationally efficient variation of  SME \cite{lu2019robust}  that shows promising convergence performance  under much less computation time. The details of fast SME are discussed in Appendix~\ref{app:numerical details} and its theoretical analysis  is left for future work.

Figure~\ref{fig:error_active} compares the four methods over the sample horizon for the pendulum example. Since the active method does not update at every step due to its planning phase, we match the average update frequency across all methods and report means and standard deviations over five trajectories. Parameter ambiguity is quantified by the normalized estimation error for LSE-type methods and by the normalized maximum coordinate width for SME-type methods. We observe that active exploration reduces the estimation error faster in the early stage and consistently improves over non-active LSE. After the initial transient, however, the SME methods achieve the smallest uncertainty width for the same sample size, with vanilla SME the tightest and Fast SME a close approximation.

Table~\ref{tab:active_nonactive_comparison} reports the number of samples and the wall-clock time required to reach a  precision level $\epsilon = 10^{-4}$, where the wall time includes both data collection and estimation. Because the SME widths decrease monotonically while the LSE-type errors need not, we use different stopping criteria: for the SME methods, the sample complexity is the first time all coordinate widths fall below $\epsilon$; for LSE and active exploration, it is the first time the normalized error falls below $\epsilon$ and remains below it for the rest of the trajectory. The table shows that the SME methods require substantially fewer samples but incur a higher per-sample cost due to the LP updates; LSE has the lowest per-update cost but requires many more samples; and active exploration improves the sample efficiency of non-active LSE, with wall time between the two.

\begin{figure}[ht]
    \centering
    \includegraphics[width=0.98\linewidth]{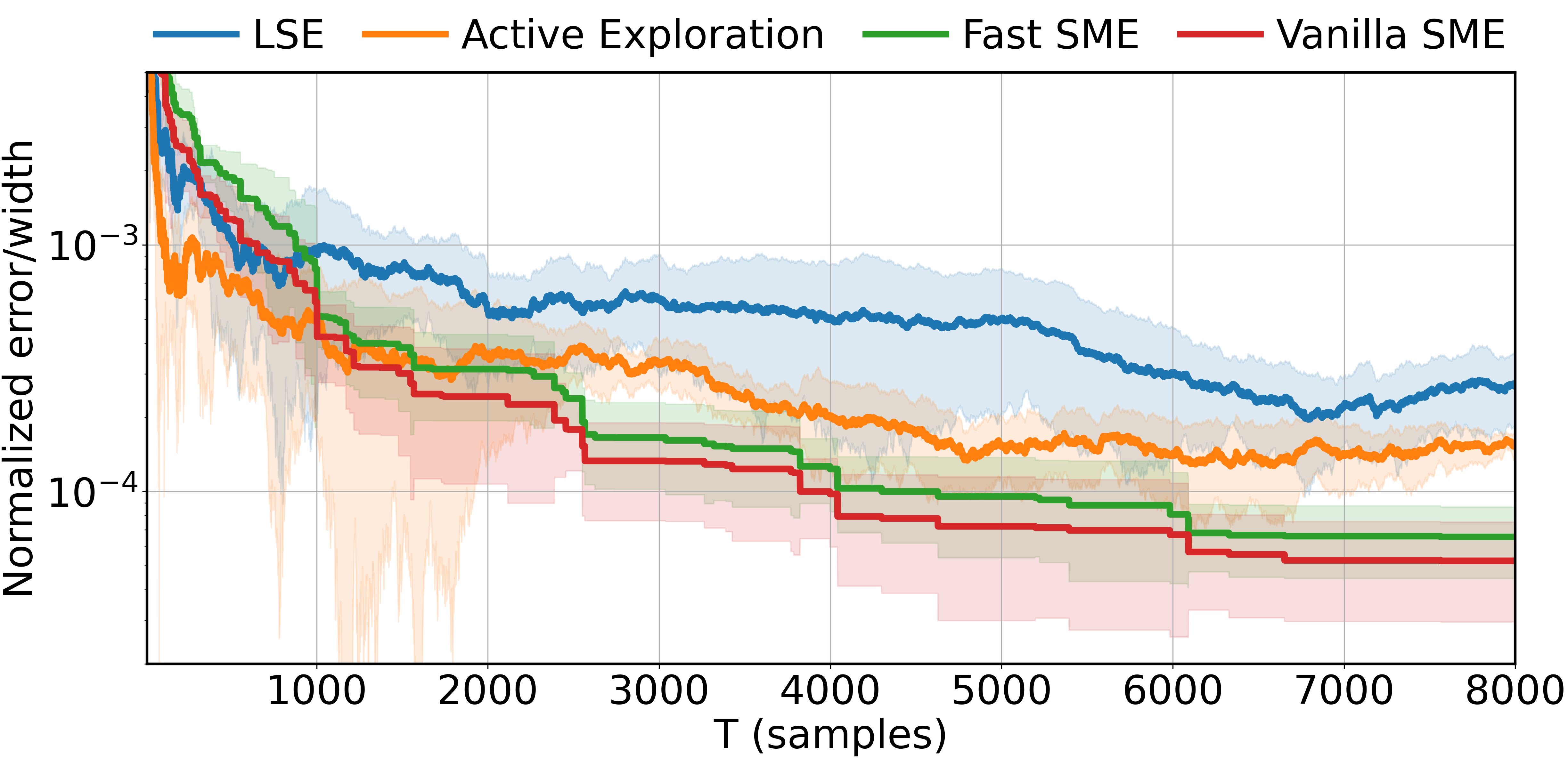}
    \caption{Normalized estimation error/width using active exploration~\cite{mania2022active} and non-active methods including least square and set membership. Exploration noise $\eta_{t}$ and system disturbances $w_{t}$ are i.i.d. generated from $\texttt{trunc-G}(0, 1, [-1, 1])$. The same $w_{t}$ sequence is used in the active exploration method. }
    \label{fig:error_active}
\end{figure}

\begin{table}[ht]
\centering
\footnotesize
\begin{tabular}{lcc}
\hline
Methods & Samples & Wall Time (s) \\
\hline
LSE
& $(1.6\pm 0.7)\times 10^{5}$
& $11.7\pm 5.2$ \\
Active Exploration
& $(4.1\pm 2.8)\times 10^{4}$
& $20.1\pm 13.5$ \\
Fast SME
& $(4.5\pm 1.4)\times 10^{3}$
& $17.2\pm 5.3$ \\
Vanilla SME
& $(3.5\pm 1.5)\times 10^{3}$
& $187.2\pm 65.8$ \\
\hline
\end{tabular}
\vspace{0.2cm}
\caption{The number of samples and wall-clock time required to reach the estimation error threshold of $10^{-4}$. Entries are reported as mean $\pm$ empirical standard deviation across five independent trajectories.} 
\label{tab:active_nonactive_comparison}
\end{table}

%% file: Conclusion.tex
\section{Conclusion}\label{sec:con}

This paper studies the identification of LPN systems with real-analytic feature functions from a single trajectory under non-active exploration. Our  results include establishing BMSB conditions for LPN systems, providing convergence rates of LSE and SME for non-active exploration, constructing counter-examples to demonstrate the importance of the assumed conditions, and numerically validating our theoretical results.

Interesting and important future directions include 1)~deriving explicit expressions for the BMSB constants $(s_\phi, p_\phi)$, whose existence is established in this work, and characterizing their dependence on system dimensions for structured sub-classes of LPN systems; 2)~extending the analysis to nonlinear systems beyond linear parameterization, such as generalized linear system; 3)~leveraging the convergence guarantees for online control designs, such as robust adaptive model predictive control; 4) analyzing convergence rates of fast SME variants, 5) analyzing  convergence rates of system identification of LPN systems under unbounded Gaussian noises, 6) exploring convergence rates for more general nonlinear systems, including non-parametric learning methods, 7) analyzing finite sample identification of LPN systems with output feedback, etc.

%% file: appendix_1.tex
\section*{Appendix}

\subsection{Proof of Corollary \ref{cor: bmsb closed loop}}\label{appendix: LSE closed loop cor}


The proof is straightforward by applying Theorems \ref{th:BMSB-o-l} and  \ref{thm: LSE open loop} to the  closed-loop system. 

With the stabilizing controller $u_t=\pi(x_t)+\eta_t$, the closed loop system is $x_{t+1}= \Theta_* \phi(x_t, \pi(x_t)+\eta_t)+w_t$. We can define new feature functions as $\tilde \phi(x_t, u_t)= \phi(x_t, \pi(x_t)+u_t)$. Notice that this equivalent LPN system 
\begin{equation}
    x_{t+1} =\Theta_* \tilde \phi(x_t, u_t)+w_t\label{equ: new closed loop}
\end{equation}
with i.i.d. random inputs $u_t=\eta_t$
satisfies the assumptions in Theorems \ref{th:BMSB-o-l} and \ref{thm: LSE open loop}. 

In particular, the feature functions $\tilde \phi$ are real-analytic because the composition of two real-analytic functions is still real-analytic by Proposition \ref{prop: real analytic properties} and both $\phi$ and $\pi$ are real-analytic by Assumptions \ref{ass:analytic} and \ref{ass: stabilizing controller}. Further, the conditions on disturbances and inputs are satisfied by Assumption \ref{ass:bounded-iid-noise}. Finally, due to the stabilizing property in Assumption \ref{ass: stabilizing controller}, the system \eqref{equ: new closed loop} is open-loop stable and thus satisfies Assumption \ref{ass:liss}. 

Therefore, we can apply Theorem \ref{th:BMSB-o-l} to establish $(1,\, \tilde s_{\phi}^{2}I_{n_{\phi}},\, \tilde p_{\phi})$-BMSB condition with $\tilde s_{\phi}>0$ and $0<\tilde p_{\phi} <1$ for the closed-loop system. Consequently, we can apply the BMSB condition above and Theorem \ref{thm: LSE open loop} to establish the boundedness of the feature functions $\|\phi(x_t,u_t)\|_2 \leq \tilde \phi_{\max}<+\infty$ and the convergence rate $O(1/\sqrt T)$ of LSE for the closed loop system.

\subsection{Semi-Continuity Property in the Proof of Lemma \ref{lem: P(Nvz)<1}}\label{appendix: semi continuity}
This subsection provides a proof for the following semi-continuity property used in the proof of Lemma \ref{lem: P(Nvz)<1}.
\begin{lemma}\label{lem: semi continuous w eta}
    If $\eta$ and $w$ follow mutually independent distributions that are semi-continuous individually, then the joint distribution of $(\eta, w)$ is also semi continuous.
\end{lemma}
\begin{proof}
This statement is intuitively correct, but a rigorous proof relies on the Lebesgue decomposition theorem below.
\begin{proposition}[Lebesgue Decomposition \cite{pollard2002user}]\label{prop: lebesgue decompose}
    Consider a probability distribution $\Pb$ on $\R^n$, there exists unique finite measures $\Pb_{ac}$ and $\Pb_{s}$ on $\R^n$ such that
    $$\Pb=\Pb_{ac}+\Pb_s$$
    where $\Pb_{ac}$ is absolutely continuous on $\R^n$, i.e. $\texttt{Leb}^{n}(E)=0$ implies $\Pb_{ac}(E)=0$ for any measurable set $E$; and $\Pb_s$ is singular on $\R^n$, i.e. there exists a set $E$ with $\texttt{Leb}^{n}(E)=0$ such that $\Pb_s(E^{\mathrm c})=0$.\footnote{$\Pb_{ac}$ and $\Pb_s$ may not be probability measures because their total masses may be smaller than 1. }
\end{proposition}

Based on Proposition \ref{prop: lebesgue decompose}, we establish the following characterization of semi-continuity.

\textit{Claim 1:} A probability measure $\Pb$ is semi-continuous if and only if $\Pb_{ac}\not=0$. 

\begin{proof} If $\Pb_{ac}=0$, then $\Pb=\Pb_s$ is concentrated on a
set $E$ with $\texttt{Leb}^{n}(E)=0$ but $\Pb(E)=1$, which violates the definition of semi-continuity. 

Conversely, if $\Pb_{ac}(\mathbb{R}^n)>0$, then for every $E$ with
$\mathrm{Leb}_n(E)=0$,
we have $
\Pb(E)=\Pb_{ac}(E)+\Pb_{s}(E)
      =\Pb_s(E)\le\Pb_{s}(\mathbb{R}^n)
      =1-\Pb_{ac}(\mathbb{R}^n)<1,
$
which proves the semi-continuity by Definition \ref{def:semi-continuous}. 
\end{proof}

\vspace{4pt}

Now, we are ready to prove Lemma \ref{lem: semi continuous w eta}. By Proposition \ref{prop: lebesgue decompose}, we decompose the semi continuous distributions $\Pb_\eta$ and $\Pb_w$ as follows:
$$\Pb_\eta=\Pb_{\eta,ac}+\Pb_{\eta,s}, \quad \Pb_w=\Pb_{w,ac}+\Pb_{w,s}.$$
Then, by the independence of $\eta$ and $w$, the joint distribution $\Pb_{\eta,w}$ can be decomposed by 
$$\Pb_{\eta,w}= \Pb_\eta\otimes \Pb_w= \Pb_{\eta,ac}\otimes \Pb_{w,ac}+ \Pb_{\eta,w,s}$$
where $\Pb_{\eta,w,s}=\Pb_{\eta,s}\otimes \Pb_{w,ac}+ \Pb_{\eta,ac}\otimes \Pb_{w,s}+ \Pb_{\eta,s}\otimes \Pb_{w,s}$ and $\otimes$ denotes the product of two measures.  By Claim 1, we have $\Pb_{\eta,ac}(\R^{n_u})>0$ and $\Pb_{w,ac}(\R^{n_x})>0$. Therefore, the joint distribution has a nonzero absolutely continuous component: $\Pb_{\eta,ac}\otimes \Pb_{w,ac}(\R^{n_u+n_x})>0$, which guarantees the semi continuity of the joint distribution by Claim 1. 
\end{proof}

%% file: appendix_2.tex
\subsection{Proofs for Section \ref{sec: SME} on SME}\label{appendix: SME}

This subsection provides the proofs for the convergence rates of SME in Section \ref{sec: SME}, including Theorem \ref{th:sme-con-o-l}, Corollaries \ref{cor:sme-rate} and \ref{cor:sme-con-c-l}, and examples of $q_w(\ell)=c_w\ell$.

\vspace{4pt}

\noindent \textbf{Proof of Theorem \ref{th:sme-con-o-l}.} The proof is straightforward by applying the BMSB condition in Theorem \ref{th:BMSB-o-l} and the SME's convergence rate for general linear regression in \cite{li2024icml}, which is reviewed in the proposition below.

\begin{proposition}[SME  for General Linear Regression~\cite{li2024icml}]\label{prop: SME general}
Consider a general linear regression $y_{t} = \Theta_{*} x_{t} + w_{t}$, where  $y_t\in \R^{n_y}$, $x_t \in \R^{n_x}$ is adapted to a natural filtration $\F_t$. Under Assumption \ref{ass:bounded-iid-noise}, Assumption \ref{ass:w-tight}, and the following conditions, if  $\{x_{t}\}_{t \geq 1}$ satisfies the $(1, s_x^2 I_{n_x}, p_x)$-BMSB condition and there exists $x_{\max} <+\infty $ such that $\|x_{t}\|_{2} \leq x_{\max}$ almost surely for all $t \geq 0$, then for any $m \geq 1$ and  $\varrho \in (0, 1)$, when $T > m$, the diameter of the SME's uncertainty set 
satisfies:
\begin{align*}
&\mathbb{P}\bigg(\textup{diam}(\ThetaSet) > \varrho\bigg)\leq 544 \frac{T}{m} n_{x}^{2.5} \log(a_{2}n_{x})a_{2}^{n_{x}}\exp(-a_{3}m)\\
    &+ 544 n_{x}^{2.5} n_{y}^{2.5} \log(a_{4} n_{x} n_{y}) a_{4}^{n_{x}n_{y}}\bigg(1-q_{w} \bigg(\dfrac{a_{1}\varrho}{4\sqrt{n_{y}}}\bigg)\bigg)^{\lceil T/ m \rceil},
\end{align*}
where $a_{1} = \frac{s_{x}p_{x}}{4},a_{2} = \frac{64x_{\max}^{2}}{s_{x}^{2}p_{x}^{2}},a_{3} = \frac{p_{x}^{2}}{8}$, and $a_{4} = \max (\frac{4x_{\max}\sqrt{n_{y}}}{a_{1}}, 1)$.
\end{proposition}

\vspace{2pt}

The remaining proof is straightforward. First,  we verify  the conditions in  Proposition \ref{prop: SME general}. Notice that our LPN system \eqref{eq:sys} can be viewed as a linear regression. Assumptions \ref{ass:bounded-iid-noise} and \ref{ass:w-tight} are assumed to hold in Theorem \ref{th:sme-con-o-l}. The BMSB condition  and the boundedness of $\phi(x_t,u_t)$ have been established in Theorem \ref{th:BMSB-o-l}. Therefore, we can prove Theorem \ref{th:sme-con-o-l} by applying Proposition \ref{prop: SME general}. 

The explicit formulas of the coefficients hidden in the $\tilde O(\cdot)$ notation are provided below.
\begin{align*}
&\mathbb{P}\bigg(\textup{diam}(\ThetaSet) > \varrho\bigg)\leq \underbrace{544 \frac{T}{m} n_{\phi}^{2.5} \log(a_{2}n_{\phi})a_{2}^{n_{\phi}}\exp(-a_{3}m)}_{\text{Term 3}}\\
    &+ \underbrace{544 n_{x}^{2.5} n_{\phi}^{2.5} \log(a_{4} n_{x} n_{\phi}) a_{4}^{n_{x}n_{\phi}}\bigg(1-q_{w} \bigg(\frac{a_{1}\varrho}{4\sqrt{n_{x}}}\bigg)\bigg)^{\lceil T/ m \rceil}}_{\text{Term 4}},
\end{align*}
where $a_{1} = \frac{s_{\phi}p_{\phi}}{4}$, $a_{2} = \frac{64\phi_{\max}^{2}}{s_{\phi}^{2}p_{\phi}^{2}}$, $a_{3} = \frac{p_{\phi}^{2}}{8}$, and $a_{4} = \max(\frac{16\phi_{\max}\sqrt{n_{x}}}{s_{\phi}p_{\phi}},1)$.

\vspace{4pt}

\noindent \textbf{Examples of $q_w(\ell)=c_w\ell$.} 
We provide explicit formulas of $q_w(\ell)$ for the  two example distributions discussed in Section \ref{sec: SME}.

When $w_{t}$ follows a uniform distribution on $[-w_{\max}, w_{\max}]^{n_{x}}$, then $q_{w}(\ell) = c_{w} \ell$ with $c_{w} = \frac{1}{2w_{\max}}$.  

When $w_{t}$ follows a truncated-Gaussian distribution on $[-w_{\max}, w_{\max}]^{n_{x}}$, generated by a Gaussian distribution with zero mean and covariance matrix $\sigma_{w}^{2}I_{n_{x}}$, then $q_{w}(\ell) = c_{w} \ell$ with 
$$c_{w} = \frac{1}{\min(\sqrt{2\pi}\sigma_{w}, 2w_{\max})} \exp(\frac{-w_{\max}^{2}}{2\sigma_{w}^{2}}).$$

\vspace{4pt}

\noindent \textbf{Proof of Corollary \ref{cor:sme-rate}.} 
We prove this corollary by
showing that $\mathbb{P}\big(\mathrm{diam}(\ThetaSet) > \varrho\big)
\leq 2\delta$. In particular, we will choose proper $m$ and $\varrho$ such that $\text{Term 3}\leq \delta$ and $\text{Term 4}\leq \delta$, which completes the proof. 

First, to ensure $\text{Term 3}\leq \delta$, we need 
\begin{align*}
    544 T n_{\phi}^{2.5} \log(a_{2}n_{\phi})a_{2}^{n_{\phi}}\exp(-a_{3}m) \leq \delta.
\end{align*}
By reorganizing the coefficients and taking logarithms on both sides, we can show that  
\begin{align*}
    m & \geq \dfrac{1}{a_{3}} \Big( \log(T/\delta) + n_{\phi} \log(a_{2})+ 2.5 \log(n_{\phi})  \\
    & \qquad \qquad \qquad \qquad \qquad + \log\log(a_2 n_\phi) + \log(544)\Big)
    \end{align*}
    can guarantee $\text{Term 3}\leq \delta$. 

Next, we show that the following choice of $\varrho$
\begin{align*}
     \varrho & = \frac{4\sqrt{n_x}\, m}{c_w a_1 T}
    \Big( \log\big(\tfrac{1}{\delta}\big) + n_x n_{\phi}\log(a_4) \notag \\
    & \qquad  + 2.5\log(n_x n_{\phi}) + \log\log(a_4 n_x n_{\phi}) + \log(544) \Big),
\end{align*}
    can guarantee $\text{Term 4}\leq \delta$. In particular, we denote $C\coloneqq 544\, n_x^{2.5} n_{\phi}^{2.5}
\log(a_4 n_x n_{\phi})\, a_4^{n_x n_{\phi}}$ for simplicity and note $C \geq 1$.
Since $q_w(\ell) = c_w \ell$, $1 - y \leq \exp(-y)$ for all
$y \in \mathbb{R}$, and since $\lceil T/m \rceil \geq T/m$, we can derive the following upper bound on $\text{Term 4}$:
\begin{align*}
    \text{Term 4}
    &= C \Big( 1 - q_w\Big( \frac{a_1 \varrho}{4\sqrt{n_x}} \Big)
       \Big)^{\lceil T/m \rceil}
    \leq C \Big( 1 - \frac{c_w a_1 \varrho}{4\sqrt{n_x}}
       \Big)^{T/m} \\
    & \leq C \exp\Big( - \frac{T}{m} \cdot
       \frac{c_w a_1 \varrho}{4\sqrt{n_x}} \Big).
\end{align*}
Then, with the $\varrho$  defined above, we have
\begin{align*}
    &\frac{T}{m} \cdot \frac{c_w a_1 \varrho}{4\sqrt{n_x}}
    = \log\big(\tfrac{1}{\delta}\big) + n_x n_{\phi}\log(a_4)
    + 2.5\log(n_x n_{\phi}) \\
    &\qquad \qquad \qquad \qquad + \log\log(a_4 n_x n_{\phi}) + \log(544)
    = \log\big(\tfrac{C}{\delta}\big).
\end{align*}
Consequently,  $\text{Term 4} \leq C \exp\big( -\log(C/\delta) \big) = \delta$.

In conclusion, we have shown 
$\mathbb{P}\big(\mathrm{diam}(\ThetaSet) > \varrho\big)
\leq 2\delta$. Equivalently, with probability at least $1 - 2\delta$, we have $\mathrm{diam}\big(\ThetaSet\big) \leq \varrho$. By substituting the definitions of $\varrho$ and $m$ above, we have completed the proof.

\vspace{3pt}

\noindent \textbf{Proof of Corollary \ref{cor:sme-con-c-l}.} 
The proof is  similar to   the proof of Theorem \ref{th:sme-con-o-l}. In particular, by Corollary \ref{cor: bmsb closed loop}, we have the BMSB condition and the boundedness of $\phi(x_t, u_t)$, so we can apply Proposition \ref{prop: SME general} to obtain the diameter bound. Then, by following the proof of Corollary \ref{cor:sme-rate}, we obtain the $\tilde O(1/T)$ convergence rate of SME for the closed-loop system.

%% file: appendix_6.tex
\subsection{More Details for Numerical Experiments}\label{app:numerical details}

\vspace{3pt}

\noindent \textbf{Numerical Estimation of Parameters in  Theoretical Bounds.}\label{append: numerical}
This appendix will discuss how to estimate the BMSB parameters numerically.  In the following, we first fix $s_{\phi} = \bar{s}$ for some $\bar s>0$ and numerically estimate $p_\phi$. If the estimation  satisfies $0<p_\phi<1$, we have completed the numerical procedures of the BMSB parameter estimation. If the estimated $p_\phi \not \in (0,1)$, we   decrease the value of $s_\phi$ and repeat the procedures, until we identified $s_{\phi}, p_{\phi}$ that satisfy the BMSB conditions. 

 The estimation of $p_{\phi}$ is discussed below. Recall that the BMSB condition implies that 
 \[
    p_{\phi} \leq \inf_{\mathcal{F}_{t}, t \geq 0}\ \inf_{\|v\|_{2}=1}\  {\mathbb{P}}\bigg(|v^{\intercal}\phi(z_{t+1})| \geq s_{\phi} \ \big|\ {\mathcal{F}}_{t}\bigg).
\]
Further, recall that $z_t \in \F_t$ and \( \phi(z_{t+1}) = \phi( \Theta_*\phi(z_{t}) + w_{t}, u_{t+1} ) \), so \( \phi(z_{t+1})\ |\ {\mathcal{F}}_{t} \) is random due to $w_t,u_{t+1}$. To estimate $p_{\phi}$, we generate independent trajectories of data $\{\{x_t^i, u_t^i, w_t^i\}_{t=1}^T\}_{i=1}^{N_d}$. In addition, we sample $N_v$ vectors uniformly from $\Sb$ and denote the set as $\Sb_d$. We  estimate $p_{\phi}$ by 
\[
\bar{p} = \min_{1\leq i \leq N_d}\min_{1\leq t \leq T}  \min_{v \in \Sb_d}\  \mathbb{P} \bigg(|v^{\intercal}\phi(\Theta_{*}\phi(z_t^i) + w_{t}^i, u_{t+1}^i)| \geq \bar{s} \bigg).
\] 
The probability is estimated by Monte Carlo simulations. In particular, we generate multiple random samples of $w_t, u_{t+1}$ given $z_t$ and the probability is estimated by the ratio between the number of samples satisfying the condition over the total samples.

\vspace{4pt}

\noindent \textbf{Fast SME.} 
 A brief explanation of fast SME used in Section \ref{sec:num-exp} is provided below.
Fast SME maintains a fixed-complexity polytopic outer approximation of the exact unfalsified set, $\widehat{\Theta}^J_t = \{\widehat{\Theta} : A^J \widehat{\Theta} \le b_t\}$, where the rows of $A^J$ are generated from a regular polytope with $J$ sides and $b_t$ is tightened at each step so that $\widehat{\Theta}^J_t$ remains an outer approximation. This is similar in spirit to the bounded-complexity set-membership updates used in robust adaptive MPC \cite{lorenzen2019robust, lu2019robust}.